\documentclass{article}
 
\usepackage{hyperref}
\usepackage[table,dvipsnames]{xcolor}
 \hypersetup{
    colorlinks=true,
    linkcolor=magenta,     
    citecolor= ForestGreen} 
   
\usepackage{enumitem,pifont,float,framed,fullpage}

\usepackage{graphicx}
\usepackage{amsmath}
\usepackage{amssymb}
\usepackage{amsthm}
\usepackage{thmtools}
\usepackage{thm-restate}

\newtheorem{theorem}{Theorem}
\newtheorem{corollary}{Corollary}
\newtheorem{lemma}{Lemma}[section]

\newtheorem{proposition}{Proposition}
\newtheorem{problem}{Problem}
\newtheorem{remark}{Remark}

\newtheorem*{theorem*}{Theorem}

\usepackage{graphicx,float}
\graphicspath{{fig/}} 

\DeclareMathAlphabet{\mathcal}{OMS}{cmsy}{m}{n} 
\DeclareMathAlphabet\mathbfcal{OMS}{cmsy}{b}{n}
 
\newcommand{\tmop}[1]{\ensuremath{\operatorname{#1}}}
\DeclareMathOperator*{\minimize}{\operatorname{minimize}}
\DeclareMathOperator*{\maximize}{\operatorname{maximize}}
\def \st {\operatorname*{subject\ to\ }}
\def \nn        {\nonumber}

\def \eye       {\mathbf{I}}
\def \zero      {\mathbf{0}}
 
\def \lg        {\langle}
\def \rg        {\rangle}
\def \diag      {\tmop{diag}}

\def \dif       {\tmop{d}}
 
 \newcommand{\bs}[1]{\boldsymbol{#1}} 
\def \valpha    {{\bs{\alpha}}}
\def \vbeta     {{\bs{\beta}}}
\def \vgamma    {{\bs{\gamma}}} 
\def \vlambda   {{\bs{\lambda}}}
\def \vtheta    {{\bs{\theta}}}
\def \mSigma    {{\bs{\Sigma}}}
 
\newcommand{\vct}[1]{\boldsymbol{#1}}
\def \a {\vct{a}}
\def \b {\vct{b}}
\def \c {\vct{c}}

\def \e {\vct{e}}

\def \y {\vct{t}}
\def \u {\vct{u}}
\def \v {\vct{v}}
\def \w {\vct{w}}
\def \x {\vct{x}}
\def \y {\vct{y}}
\def \z {\vct{z}}

\newcommand{\mtx}[1]{\mathbf{#1}}
\def \mA {\mtx{A}}
\def \mB {\mtx{B}}
\def \mC {\mtx{C}}

\def \mH {\mtx{H}}
\def \mI {\mtx{I}}

\def \mP {\mtx{P}}
\def \mQ {\mtx{Q}}

\def \mU {\mtx{U}}
\def \mV {\mtx{V}}
\def \mW {\mtx{W}}
\def \mX {\mtx{X}}

\newcommand{\tensor}[1]{\mathbfcal{#1}}

\def \tQ {\tensor{Q}}
\def \tR {\tensor{R}}
\def \tT {\tensor{T}}

\def \calF {\mathcal{F}}

\def \calL {\mathcal{L}}
\def \calM {\mathcal{M}}
\def \calN {\mathcal{N}}

\def \bbK {\mathbb{K}}
\def \bbN {\mathbb{N}}
\def \bbS {\mathbb{S}} 

\def \S {{\mathbb{S}}}
\def \R {{\mathbb{R}}}

\def \A {{\mathcal{A}}}

\begin{document}

\title{A Super-Resolution Framework for Tensor Decomposition}

\author{Qiuwei Li\thanks{Alibaba DAMO Academy, Bellevue, WA}      
\and
Ashley Prater\thanks{Force Research Laboratory, NY}
\and
Lixin Shen\thanks{Department of Mathematics, Syracuse University, NY}
\and
Gongguo Tang\thanks{Department of Electrical, Computer \& Energy Engineering, University of Colorado Boulder, CO}
}

\maketitle 

\begin{abstract}
{This work considers a super-resolution framework for overcomplete tensor decomposition. Specifically, we view tensor decomposition as a  super-resolution problem of recovering a sum of Dirac measures on the sphere and solve it by minimizing a continuous analog of the $\ell_1$ norm on the space of measures.  The optimal value of this optimization defines the tensor nuclear norm. Similar to the separation condition in the super-resolution problem, by explicitly constructing a dual certificate,  we develop incoherence conditions of the tensor factors so that they form the unique optimal solution of the continuous analog of $\ell_1$ norm minimization. Remarkably, the derived incoherence conditions are satisfied with high probability by random tensor factors uniformly distributed on the sphere, implying global identifiability of random tensor factors.}
{Atomic norm minimization,  Dual certificate, Nonconvex, Tensor decomposition, Tensor nuclear norm, Super-resolution}
\end{abstract}

\section{Introduction}\label{sec:intro}
Tensors provide natural representations for massive multi-mode datasets encountered in  many applications including  image and video processing \cite{Liu:2013bh}, collaborative filtering \cite{Chen:2005jn}, array signal processing \cite{Lim:2010if}, convolutional networks design \cite{ge2017learning} and psychometrics \cite{Smilde:2005wf}.
Tensor methods also form the backbone of many machine learning, signal processing, and statistical algorithms, including independent component analysis (ICA) \cite{Cardoso:1989eo}, latent graphical model learning \cite{anandkumar2017analyzing}, dictionary learning \cite{Barak:2014vw}, and Gaussian mixture estimation \cite{Hsu:2013iz}.	The utility of tensors in such diverse applications is mainly due to the ability to identify {\em overcomplete}, {\em non-orthogonal} factors from tensor data as already suggested by \textsf{Kruskal's theorem} \cite{Kruskal:1977dx}. This is
known as tensor decomposition, which describes the problem of decomposing a tensor into a linear combination of a small number
of rank-1 tensors. The identifiability of  tensor factors is in sharp contrast to the inherent ambiguous nature of matrix decompositions without additional assumptions such as orthogonality and non-negativity.

In addition to its practical applicability, tensor decomposition is also of fundamental theoretical interest in solving linear inverse problems involving low-rank tensors. For one thing, theoretical results for tensor decomposition inform what types of rank-1 tensor combinations are identifiable given full observations. For another, a dual polynomial is constructed to certify a particular decomposition, which is useful in investigating the regularization power of the tensor nuclear norm for tensor inverse problems, including tensor completion, tensor denoising, and robust tensor principal component analysis. We expect that the {\em dual certificate} constructed in this work will play an important role in these tensor inverse problems similar to  that of  the subdifferential characterization of matrix nuclear norm in matrix completion and low-rank matrix recovery \cite{Candes:2009kja, Recht:2010hta}.

\subsection{The Tensor Decomposition Problem}
In this work, we focus on third-order  nonsymmetric tensors that can be decomposed  into a linear combination of unit-norm, rank-1  tensors  of the form $\u\otimes\v\otimes\w\in\R^{n_1}\otimes\R^{n_2}\otimes\R^{n_3}$. More precisely, consider the following nonsymmetric tensor decomposition
\begin{align}\label{eqn:true}
\tT=\sum_{p=1}^r  \lambda_p^\star \u_p^\star\otimes \v_p^\star \otimes \w_p^\star.
\end{align}
Through this work, we assume the rank-1 tensor factors $\{(\u_p^\star,\v_p^\star,\w_p^\star)\}$ are living on the unit spheres and might be {\em overcomplete}, that is,   $r$  is potentially greater than the individual tensor dimensions $n_1, n_2$ and $n_3$. Without loss of generality, we assume that the coefficients $\{\lambda_p^\star\}$  are positive as their signs can be absorbed into the factors.

\begin{problem}
The tensor decomposition problem is the inverse problem of retrieving those ground-truth rank-1 tensor factors $\{(\u_p^\star,\v_p^\star,\w_p^\star)\}_{p=1}^r$ from the tensor data $\tT$ in \eqref{eqn:true} \cite{li2015overcomplete}.	
\end{problem}
 
\subsection{The Super-Resolution Framework}
Tensor decomposition is  an extremely challenging problem \cite{Hillar:2013by}.  This is because we lack proper theories for basic tensor concepts and operations such as singular values, vectors, and singular value decompositions. To address these challenging issues, we will consider a \emph{super-resolution} framework for tensor decomposition. More precisely, we can view tensor decomposition as a problem of \emph{measure estimation} from moments. This is because we can rewrite the tensor decomposition \eqref{eqn:true} as a integral on the unit spheres  $\bbK:=\S^{n_1-1}\times\S^{n_2-1}\times\S^{n_3-1}$:
\begin{align}
\tT=\int_{\mathbb{K}} \u\otimes\v\otimes\w  \dif \mu^\star.
\end{align}
and then the problem of retrieving the rank-1 tensor factors $\{(\u_p^\star,\v_p^\star,\w_p^\star)\}$ from the observed tensor entries in $\tT$ is equivalent to recovering a linear combination of Dirac measures defined on the unit spheres $\bbK$:
\begin{align}\label{eqn:true_measure}
\mu^\star = \sum_{p=1}^r \lambda_p^\star \delta(\u-\u_p^\star,\v-\v_p^\star,\w-\w_p^\star)
\end{align}

Several advantages are offered by this super-resolution framework. First, it provides a natural way to extend the $\ell_1$ norm minimization in finding sparse representations for finite dictionaries~\cite{spark} to tensor decomposition.   By viewing the set of rank-1 tensors $\A = \{\u\otimes \v \otimes \w: (\u,\v,\w)\in\bbK\}$ as a dictionary with an infinite number of atoms, this formulation allows   us  to find a \emph{sparse}\footnote{The decomposition \eqref{eqn:true} is sparse, because in most practical scenarios,  $r$ is much smaller than the product $n_1n_2n_3$.
}  representation of $\tensor{T}$  by minimizing the $\ell_1$ norm   of the representation coefficients  with respect to the  dictionary $\A$. More precisely, we recover $\mu^\star$ from the tensor $\tensor{T}$ by solving a continuous analog of $\ell_1$ norm minimization (a.k.a. the total mass minimization over the space of measures)
\begin{align}
\minimize_{\mu \in \calM_+(\bbK)} \mu(\bbK) \ \st \tensor{T} = \int_\bbK \u \otimes \v \otimes \w \dif \mu \label{eqn:method}
\end{align}
where $\mathcal{M}_+ ( \mathbb{K} )$ is the set of (nonnegative) Borel measures on $\mathbb{K}$, and $\mu(\mathbb{K})$ is the total measure/mass of the set $\mathbb{K}$  measured by the Borel measure $\mu \in \calM_+(\bbK)$. Second, the optimal value of the total mass minimization defines precisely the {\em tensor nuclear norm} \cite[Proposition 3.1]{friedland2016nuclear}, which is a special case of atomic norms \cite[Eq. (2)]{Chandrasekaran:2010hl} corresponding to the atomic set $\A$. The tensor nuclear norm is useful in many tensor inverse problems,  such as, tensor completion \cite{cai2019nonconvex} and robust tensor principal component analysis \cite{lu2016tensor}.

\section{Main Results}
The main focus of this work is on characterizing the conditions when
the tensor factors   $\{(\u_p^\star, \v_p^\star, \w_p^\star)\}_{p=1}^r$ correspond to the unique optimal solution of the continuous analog of $\ell_1$ norm minimization \eqref{eqn:method}, which is extension of the incoherence condition in matrix completion problem~\cite{Candes:2009kja}, the minimum separation condition in mathematical super-resolution \cite{Candes:2014br}, and the wrap-around distance condition in line spectral estimation \cite{Tang:2013fo}. More precisely, we develop the following three assumptions, namely, incoherence condition, bounded spectral norm condition, and Gram isometry condition. For ease of exposition, in what follows, these assumptions and the main result of this work will be presented for square tensors with $n_1 = n_2 = n_3 = n$.

\medskip
\noindent{\bf Assumption I: Incoherence condition.} 
\begin{align}\label{cond:incoherence}
\Delta&:=\max_{p \neq q} \max\{ | \langle \u_{p}^{\star} ,\u_{q}^{\star} \rangle | , |
\langle \v_{p}^{\star} ,\v_{q}^{\star} \rangle | , | \langle \w_{p}^{\star}
,\w_{q}^{\star} \rangle | \}
\leq \frac{\tau(\log n)}{\sqrt{n}},
\end{align}
where $\tau(\cdot)$ is a polynomial function of its argument.\footnote{That being said, $\tau(\cdot)$ is of the form $\tau(x)=a_{m} x^{m}+a_{m-1} x^{m-1}+\ldots+a_{2} x^{2}+a_{1} x+a_{0}$ with some (positive) real numbers $a$'s being the coefficients of the polynomial and some positive integer $m$ being the degree of the polynomial.}

\medskip
\noindent{\bf Assumption II: Bounded spectral norm condition.} 
 \begin{align}\label{cond:spectralnorm}
\max \{ \| \mU \| , \| \mV \| , \| \mW \| \} \leq 1+c \sqrt{\frac{r}{n}}
\end{align}
for some constant $c>0$, where $\mU:= \begin{bmatrix}\u_{1}^{\star} & \cdots  & \u_{r}^{\star}\end{bmatrix}$, $\mV:= \begin{bmatrix}\v_{1}^{\star} & \cdots & \v_{r}^{\star}\end{bmatrix}$, and $\mW:= \begin{bmatrix}\w_{1}^{\star}  & \cdots & \w_{r}^{\star}\end{bmatrix}$.

\medskip
\noindent{\bf Assumption III: Gram isometry condition.} 
\begin{align}\label{cond:gram}
\|( \mU^\top \mU ) \odot ( \mV^\top \mV )-\eye\| \leq \kappa(\log n)\frac{\sqrt{r}}{n},
\end{align}
where $\kappa(\cdot)$ is a polynomial function of its argument.
Similar bounds hold for $\mU, \mW$, and $\mV, \mW$.
 
\begin{theorem}\label{thm:main}
Suppose the target tensor $\tT\in\R^{n\times n\times n}$ admits a decomposition \eqref{eqn:true} with the normalized tensor factors $\{(\u^\star_p,\v^\star_p,\w^\star_p)\}_{p=1}^r$  satisfying Assumptions I, II, III and 
\begin{align}\label{bound:r}
r\leq \frac{n^{17/16}}{32c^2\sqrt{15\tau(\log n)}}
\end{align}
with the polynomial $\tau(\cdot)$ given in \eqref{cond:incoherence}, the constant $c$ given in \eqref{cond:spectralnorm}, and $n$ being large enough.
Then the true factors $\{(\u^\star_p,\v^\star_p,\w^\star_p)\}_{p=1}^r$ correspond to  the unique optimal solution of the continuous analog of $\ell_1$ norm minimization \eqref{eqn:method} up to a sign ambiguity.
\end{theorem}
A few remarks follow. Firstly, since $r=O\big( {n^{17/16}}/{\sqrt{\tau(\log n)}}\big)\gg n$, total mass minimization is guaranteed to recover \emph{overcomplete} tensor decompositions. Secondly, the incoherence condition is reasonable as we argue in the following. Tensor decomposition using total mass minimization is an atomic decomposition problem. The latter determines the conditions under which a decomposition in terms of atoms in an atomic set $\A$ achieves the corresponding atomic norm. For example, the singular value decomposition is an atomic decomposition for the set of unit-norm, rank-one matrices. Finally, if the incoherence bound in Assumption I is further strengthened to $O(\frac{1}{n\alpha(\log n)})$ for some polynomial $\alpha(\cdot)$, then  Assumptions II and III are consequences of Assumption I. So if the rank-one factors of an overcomplete tensor are incoherent enough, without needing  Assumptions II and III, its CP decomposition can always be uniquely identified.

We note that Assumptions I, II and III hold with high probability if the  tensor factors are generated independently according to uniform distributions on the unit spheres \cite[Lemmas 25, 31]{anandkumar2015learning}.

\begin{corollary}\label{cor:main}
If  the true tensor factors $\{(\u^\star_p,\v^\star_p,\w^\star_p)\}_{p=1}^r$ in \eqref{eqn:true} are uniformly distributed on the unit spheres, and if $r$ satisfies \eqref{bound:r}, then with high probability, 
the true tensor factors correspond to the unique optimal solution of the continuous analog of $\ell_1$ norm minimization \eqref{eqn:method}  up to a sign ambiguity.
\end{corollary}

\section{Prior Art and Inspirations}\label{sec:prior:art}
Despite the advantages provided by tensor methods in many applications, their widespread adoption has been slow due to inherent computational intractability.  Although the decomposition~\eqref{eqn:true} is a multi-mode generalization of the singular value decomposition for matrices,  extracting the decomposition from a given tensor is a nontrivial problem that is still under active investigation (cf.~\cite{comon:2009ur,kolda2009tensor}).  Indeed, even determining the rank of a third-order tensor is an NP-hard problem~\cite{Hillar:2013by}. A common strategy used to compute a tensor decomposition is to apply an alternating minimization scheme.  Although efficient, this approach has the drawback of not providing global convergence guarantees~\cite{comon:2009ur}.   Recently, an approach combining alternating minimization with power iteration has gained popularity due to its ability to guarantee the tensor decomposition results  under certain assumptions~\cite{anandkumar2015learning,Jain:2014wm}.

Tensor decomposition is a special case of atomic decomposition which is to determine when a decomposition with respect to some given atomic set $\A$
achieves the atomic norm \cite{Chandrasekaran:2010hl}. For finite atomic sets, it is now well-known that if the atoms satisfy certain  conditions such as the restricted isometry property, then a sparse decomposition achieves the atomic norm \cite{Candes:2008wp}. For the set of rank-1, unit-norm matrices, the atomic norm (the matrix nuclear norm), is achieved by orthogonal decompositions \cite{Recht:2010hta}.
When the atoms are complex sinusoids parameterized by the frequency, Cand\`es and Fernandez-Granda showed that atomic decomposition is solved by atoms with well-separated frequencies \cite{Candes:2014br}. Similar separation conditions also show up when the atoms are translations of a known waveform \cite{Tang:2013:ad, chi2014compressive,eftekhari2019sparse}, spherical harmonics \cite{Bendory:2014tl}, and radar signals parameterized by translations and modulations \cite{Heckel:2014ts}. Tang and Shah in \cite{Tang:2015gt} employed the same atomic norm idea but focused on symmetric tensors. In addition, the result of \cite{Tang:2015gt} does not apply to overcomplete decompositions. Under a set of conditions, including the incoherence condition ensuring the separation of  tensor factors, this work characterizes a class of {\em nonsymmetric} and {\em overcomplete} tensor
decompositions that achieve the tensor nuclear norm $\|\tT\|_\ast$.

Another closely related line of work is matrix recovery \cite{davenport2016overview} and tensor recovery. Low-rank matrix recovery  based on the idea of
nuclear norm minimization has received a great deal of attention in recent years \cite{Candes:2009kja, Recht:2010hta, Recht:2011up}. A direct generalization of this approach to tensors would have been using tensor nuclear norm to perform low-rank tensor recovery. However, this approach was not pursued due to the NP-hardness of computing the tensor nuclear norm \cite{Hillar:2013by} and the lack of analysis tools for tensor problems. The mainstream tensor recovery approaches are based on various forms of matricization  \cite{Gandy:2011bh, Liu:2013bh, Mu:2013tz}. Alternating minimization can also be applied to tensor recovery with performance guarantees established in recent work \cite{Huang:2014vc}. More recently, gradient descent with a good initialization is applied to the noisy symmetric tensor completion and achieves near-optimal statistical guarantees \cite{cai2019nonconvex}. {Note that all the above mentioned works study the low-rank tensor recovery problems, i.e., the number of rank-1 tensor factors is less than the factor size $n$.  While in general calculating tensor decomposition is NP-hard, the theoretical computer science community has developed some interesting algorithms for overcomplete tensor decomposition. For example, Anandkumar et al. \cite{anandkumar2015learning,anandkumar2015learning,anandkumar2017analyzing} apply the iterative power method with good initialization to the overcomplete tensor decomposition problem and provide guarantees for the linear-overcomplete case (i.e., $r\le \beta n$). In addition to these local search algorithms such as gradient descent, power method, and alternating minimization,  another line of algorithms for overcomplete tensor decomposition are based on the sum-of-squares (SoS) semidefinite programming (SDP)
hierarchy \cite{hopkins2019robust,ma2016polynomial,potechin2017exact}. Although the SoS relaxation approaches provide provable guarantees for overcomplete tensor
decomposition, they are essentially SDPs, which is not scalable to high-dimensional tensors.}

In contrast, we expect that the atomic norm, when specialized to tensors,  will achieve the information theoretical limit for tensor completion as it does for compressive sensing, matrix completion \cite{Recht:2011up,davenport2016overview}, and line spectral estimation with missing data \cite{Tang:2013fo}. Given a set of atoms, the atomic norm is an abstraction of $\ell_1$-type regularization that favors simple models. 
Using the notion of descent cones, Chandrasekaran et al. in \cite{Chandrasekaran:2010hl} argued that the atomic norm is the best possible convex proxy for recovering simple models. Particularly, atomic norms are shown in many problems beyond compressive sensing and matrix completion to be able to recover simple models from minimal number of linear measurements. For example, when specialized to the atomic set formed by complex exponentials, the atomic norm can recover signals having sparse representations in the continuous frequency domain with the number of measurements approaching the information theoretic limit without noise \cite{Tang:2013fo}, as well as achieving near minimax denoising performance \cite{Tang:2013gd}. Continuous frequency estimation using the atomic norm is also an instance of measure estimation from (trigonometric) moments.

\section{Tensor Decomposition, Atomic Norms, and Duality}\label{sec:model}
In this work, we view tensor decomposition in the frameworks of both atomic norms and measure estimation. The unit sphere of $\R^n$ is denoted by $\S^{n-1}$, and the direct product of three unit spheres $\S^{n-1}\times \S^{n-1} \times \S^{n-1}$ by $\bbK$. The tensor atomic set is denoted by  $\A = \{\u\otimes \v \otimes \w: (\u,\v,\w)\in\bbK\}$ parameterized by the set $\bbK$, where $\u\otimes \v \otimes \w$ is a rank-1 tensor  with the $(i,j,k)$th entry being $u_i v_j w_k$.  For any tensor $\tT$, its  atomic norm with respect to $\A$ is defined by \cite[Eq. (2)]{Chandrasekaran:2010hl}
\begin{align}\label{eqn:atomic:tensor}
\|\tT\|_\A
& = \inf \{ t: \tT \in t \operatorname*{conv}(\A)\}\nonumber\\
& = \inf \bigg\{ \sum_p \lambda_p: \tT = \sum_{p} \lambda_p \u_p\otimes \v_p \otimes \w_p,\lambda_p > 0,
(\u_p, \v_p, \w_p) \in \bbK\bigg\},
\end{align}
where $\operatorname*{conv}(\A)$ is the convex hull of the atomic set $\A$, and a scalar multiplying a set scales every element in the set. Therefore, the tensor atomic norm is the minimal $\ell_1$ norm of  its expansion coefficients among all valid expansions in terms of unit-norm, rank-1 tensors. The atomic norm $\|\tT\|_\A$ defined in \eqref{eqn:atomic:tensor} is also called the tensor nuclear norm and denoted by $\|\tT\|_*$ in \cite[Eq. (2.7)]{friedland2016nuclear}. We will use these two names and notations interchangeably in the following. The way of defining the tensor nuclear norm is precisely the same as that of defining the matrix nuclear norm.

We argue that the two lines in the definition \eqref{eqn:atomic:tensor} are consistent and are also equivalent to \eqref{eqn:method} as follows. Since $\operatorname*{conv}(\A) = \{\tT: \tT = \int_\bbK \u\otimes \v\otimes \w \dif \mu, \mu \in \calM_+(\bbK), \mu(\bbK) \leq 1\},$ the first line in the definition \eqref{eqn:atomic:tensor} implies that $\|\tT\|_\A$ is equal to the optimal value of \eqref{eqn:method}. Compared with the measure optimization \eqref{eqn:method}, the feasible region of the minimization defining the atomic norm in the second line of \eqref{eqn:atomic:tensor} is restricted to discrete measures. However, these two optimizations share the same optimal value as a consequence of \textsf{Carath\'eodory's convex hull theorem}, which states that if a point $\x \in \R^d$ lies in the convex hull of a set, then $\x$ can be written as a convex combination of at most $d+1$ points of that set \cite[Theorem 2.3]{Barvinok:2002vr}. Since $\tT \in \|\tT\|_\A \operatorname*{conv}(\A) =  \operatorname*{conv}(\|\tT\|_\A \A)$, $\tT$ can be expressed as a convex combination of at most $n^3 + 1$ points of the set $\|\tT\|_\A \A$, implying that the optimal value is achieved by a discrete measure with support size at most $n^3 + 1$. This argument establishes that the two lines in \eqref{eqn:atomic:tensor} as well as the measure optimization \eqref{eqn:method} are equivalent. Therefore, the atomic norm framework and the measure optimization framework are two different formulations of the same problem, with the former setting the stage in the finite dimensional space and the latter in the infinite-dimensional space of measures.

Given an abstract atomic set, the problem of atomic decomposition seeks the conditions under which a decomposition in terms of the given atoms achieves the atomic norm. In this sense, the tensor decomposition considered in this work is an atomic decomposition problem.

\subsection{Duality}\label{sec:duality}
Duality plays an important role in analyzing atomic tensor decomposition. We again approach duality from both perspectives of atomic norms and measure estimation.

First, we find the dual problem of the optimization problem {\eqref{eqn:method}}. Given $\tQ,\tT \in \R^{n \times n \times n}$, we define the tensor inner product $\langle \tQ,\tT \rangle := \sum_{i,j,k} Q_{i j k} T_{i j k}$. Standard Lagrangian analysis shows that the dual problem of {\eqref{eqn:method}} is the following semi-infinite program,
which has an infinite number of constraints:
\begin{align}
\operatorname*{maximize}_{\tQ\in \R^{n\times n \times n}}  &\ \langle \tQ,\tT \rangle \nonumber\\
\tmop{subject}   \tmop{to}  &\ \langle \tQ,\u \otimes \v \otimes \w \rangle \leq
1, \forall ( \u,\v,\w ) \in \bbK \label{eqn:dual}
\end{align}
The polynomial $q(\u, \v, \w):=\langle \tQ,\u\otimes \v\otimes \w\rangle=\sum_{i,j,k}Q_{ijk}u_iv_jw_k$ corresponding to a dual feasible solution $\tQ$ of \eqref{eqn:dual} is called a dual polynomial. The dual polynomial associated with an optimal dual solution can be used to certify the optimality of a particular decomposition, as demonstrated by the following proposition.


\begin{proposition}\label{pro:bip}
Suppose the set of rank-1 tensors $\{ \u_{p}^{\star} \otimes \v_{p}^{\star} \otimes \w_{p}^{\star}\}_{p=1}^r$   given in \eqref{eqn:true} are linearly independent. If there exists a dual solution
$\tQ \in \R^{n \times n \times n}$ to \eqref{eqn:dual} such that the
corresponding dual polynomial $q:\bbK\to \R$
\begin{align}
q ( \u,\v,\w )  &:=  \langle \tQ,\u \otimes \v \otimes \w \rangle \label{eqn:dual_poly}
\end{align}
satisfies the following \emph{Boundedness and Interpolation Property (\textsf{BIP}):}
\begin{subequations}\label{eqn:bip}
	\begin{align}
	q ( \u_{p}^{\star} ,\v_{p}^{\star} ,\w_{p}^{\star} )   &=  1\ \tmop{ for } p\in[r] \  \textsf{(Interpolation)} \label{eqn:bip:i} \\
	q ( \u,\v,\w )   &< 1 \  \tmop{ in } {\bbK}\setminus{S^\star}\ \textsf{(Boundedness)} \label{eqn:bip:b}
	\end{align}
\end{subequations}
where $[r]:=\{1,\ldots,r\}$ and
\begin{align}\label{eqn:S:star}
S^\star:=\{(a_p\u_p^\star,b_p\v_p^\star,c_p\w_p^\star): &|a_p|=|b_p|=|c_p|= a_pb_pc_p=1, p\in[r]\},
\end{align}
then $\mu^\star$ given in \eqref{eqn:true_measure} is the unique optimal solution to  \eqref{eqn:method} up to sign ambiguity.
\end{proposition}
\begin{proof}\ \
In view of \eqref{eqn:dual}, any $\tQ$ that satisfies the \textsf{BIP} in \eqref{eqn:bip} is a dual feasible solution.    We also have
\begin{align*}
\langle \tQ,\tT \rangle  & =  \left\langle \tQ, \sum_{p=1}^{r} \lambda_{p}^{\star} \u_{p}^{\star}
\otimes \v_{p}^{\star} \otimes \w_{p}^{\star} \right\rangle
=  \sum_{p=1}^{r} \lambda_{p}^{\star} \langle \tQ,\u_{p}^{\star} \otimes
\v_{p}^{\star} \otimes \w_{p}^{\star} \rangle
=  \sum_{p=1}^{r} \lambda_{p}^{\star} q ( \u_{p}^{\star} ,\v_{p}^{\star}
,\w_{p}^{\star} )
=  \mu^\star (\bbK)
\end{align*}
establishing a zero-duality gap of the primal-dual feasible solution $( \mu^{\star} ,\tQ )$. As a consequence, $\mu^{\star}$ is a primal optimal solution to \eqref{eqn:method} and $\tQ$ is a dual optimal solution to \eqref{eqn:dual}.

For uniqueness, suppose $\hat{\mu}$ is another primal optimal solution to \eqref{eqn:method}. If $\hat{\mu}
( \bbK\setminus S^\star ) >0$, then
\begin{align*}
\mu^\star(\bbK)
= \langle \tQ,\tT \rangle
= \left\langle \tQ, \int_\bbK \u \otimes \v
\otimes \w \dif  \hat{\mu} \right\rangle
<
\hat{\mu} (S^\star )  + \int_{\bbK\setminus S^\star} 1\dif \hat{\mu}
= \hat{\mu} (\bbK)
\end{align*}
contradicting the optimality of $\hat{\mu}$. So all optimal solutions
are supported on $S^\star$. To remove the sign ambiguity, we can assume an optimal solution is supported on $\{ \u_{p}^{\star} \otimes \v_{p}^{\star} \otimes
\w_{p}^{\star} \}_{p=1}^r$. Since $\{ \u_{p}^{\star} \otimes \v_{p}^{\star} \otimes
\w_{p}^{\star} \}_{p=1}^r$  are linearly independent by assumption, the coefficients $\lambda_p^\star$ can be uniquely determined from solving the linear system of equations encoded in  $T = \sum_{p=1}^r \lambda_p^\star \u_{p}^{\star} \otimes \v_{p}^{\star} \otimes \w_{p}^{\star}$. This proves the uniqueness (up to sign ambiguity).
\end{proof}

\subsection{Dual Certificate and Subdifferential}

The dual optimal solution $\tQ$ satisfying the \textsf{BIP}  is called a {\em dual certificate}, which is used frequently as the starting point to derive several atomic decomposition and super-resolution results \cite{Candes:2014br,Tang:2013fo, Tang:2015gt, Bendory:2014tl}. In Section \ref{sec:proof}, we will explicitly construct a {\em dual certificate} to prove Theorem \ref{thm:main}. In this subsection, we will relate the {\em dual certificate} with the subdifferential of the tensor nuclear norm.

First, the dual norm of the tensor nuclear norm, \emph{i.e.}, the tensor spectral norm, of a tensor $\tQ$ is given by
\begin{align}
\|\tQ\| := \sup_{\tT: \|\tT\|_* \leq 1} \langle \tQ, \tT\rangle =  \sup_{(\u, \v, \w) \in\bbK} \langle \tQ, \u\otimes \v \otimes \w\rangle.
\end{align}
The equality is due to the fact that the atomic set $\A$ are the extreme points of the unit nuclear norm ball $\{\tT: \|\tT\|_* \leq 1\}$. In light of the spectral norm definition, we rewrite the dual problem \eqref{eqn:dual} as
\begin{align}
\operatorname*{maximize}_{\tQ\in \R^{n\times n \times n}} \ \langle \tQ,\tT \rangle \ \tmop{subject} \tmop{to} \ \|\tQ\| \leq 1
\end{align}
which is precisely the definition of the dual norm of the tensor spectral norm, \emph{i.e.}, the tensor nuclear norm.

The subdifferential (the set of subgradients) of the tensor nuclear norm is defined by \cite[Definition B.20]{foucart2013mathematical}
\begin{align}
\partial \|\cdot\|_*(\tT) =\{&\tQ \in \R^{n\times n \times n}:\|\tR\|_* \geq \|\tT\|_*+\langle \tR - \tT, \tQ\rangle,
\text{for all } \tR \in \R^{n\times n \times n}\},
\end{align}
which has an equivalent representation \cite[Section 1]{Watson:1992ha}
\begin{align}\label{eqn:subdifferential}
\partial \|\cdot\|_*(\tT) = \left \{\tQ \in \R^{n\times n \times n}: \|\tT\|_* = \langle \tQ, \tT\rangle, \|\tQ\| \leq 1 \right\}.
\end{align}

For $\tT$ having an atomic decomposition given in \eqref{eqn:true}, it can be established that the defining properties of subdifferential \eqref{eqn:subdifferential} are equivalent to
\begin{subequations}
\label{eqn:subgradient:condition}
\begin{align}	\label{eqn:subgradient:condition:a}
\langle \tQ, \u_{p}^{\star}\otimes\v_{p}^{\star}\otimes\w_{p}^{\star}\rangle & =   1,\ \tmop{for} p\in[r]\\
\langle \tQ, \u\otimes \v\otimes \w \rangle & \leq  1 {} ,\ \tmop{for}   ( \u,\v,\w ) \in \bbK
\label{eqn:subgradient:condition:b}
\end{align}
\end{subequations}
We recognize that the \textsf{BIP} \eqref{eqn:bip} is a strengthened version of the subdifferential conditions \eqref{eqn:subgradient:condition}. Therefore, a {\em dual certificate}, {\em i.e.,} any $\tQ$ satisfying the \textsf{BIP}, is an element of the subdifferential $\partial \|\cdot\|_*(\tT)$. The \textsf{BIP} in fact means that $\tQ$ is an interior point of $\partial \|\cdot\|_*(\tT)$. Our proof strategy for Theorem \ref{thm:main} is to construct such an interior point in Section \ref{sec:proof}. This is in contrast to the matrix case, for which we have an explicit characterization of the entire subdifferential of the nuclear norm using the singular value decomposition (more explicit than the one given in \eqref{eqn:subdifferential}). More specifically, suppose $\mX = \mU\mSigma \mV^\top$ is the (compact) singular value decomposition of $\mX\in \R^{m \times n}$ with $\mU \in \R^{m\times r}, \mV \in \R^{n\times r}$ and ${\mSigma}$ being an $r\times r$ diagonal matrix. Then the subdifferential of the matrix nuclear norm at $\mX$ is given by \cite[Eq. (2.9)]{Recht:2010hta}
\begin{align*}
\partial \|\cdot\|_*(\mX)= \{& \mU\mV^\top + \mW: \mU^\top \mW=\zero, \mW\mV=\zero, \|\mW\| \leq 1\}.
\end{align*}
It is challenging to obtain such a characterization for tensors unless the tensor admits an orthogonal decomposition.

\subsection{Extension: Regularization Using  Tensor Nuclear Norm}\label{sec:tensor:nuc:norm}
Independent from practical considerations, we investigate tensor decomposition  for theoretical reasons. Similar to regularizing matrix inverse problems using the matrix nuclear norm, the tensor nuclear norm can be used to regularize tensor inverse problems. Suppose we observe an unknown low-rank tensor $\tT^\star$ through the linear measurement model $\y = \mathcal{B} (\tT^\star)$, we would like to recover the tensor $\tT^\star$ from the observation $\y$. For instance, when $\mathcal{B}$ samples the individual entries of $\tT^\star$, we are looking at a tensor completion problem. {Remarkably, Yuan and Zhang exploited the tensor nuclear norm approach to tensor completion and improved the state-of-the-art sample complexity in the seminal work \cite{yuan2016tensor}.} We propose recovering $\tT^\star$ by solving
\begin{equation}
\minimize_{\tT \in \R^{n\times n \times n}} \|\tT\|_*\ \st \y = \mathcal{B}(\tT)
\end{equation}
which favors a low-rank solution.  To establish recoverability, we can construct a {\em dual certificate} $\tQ$ of the form $\mathcal{B}^*(\vlambda)$, whose corresponding dual polynomial satisfies the \textsf{BIP}. Here $\mathcal{B}^*$ is the adjoint operator of $\mathcal{B}$. When the operator $\mathcal{B}$ is random, the concentration of measure guarantees that we can construct a {\em dual certificate} $\mathcal{B}^*(\vlambda)$ that is close to the one constructed in the full data case.  This fact can then be exploited to verify the \textsf{BIP} of $\mathcal{B}^*(\vlambda)$ and to establish exact recovery. When the atoms are complex exponentials parameterized by continuous frequencies, this strategy is adopted to establish the compressed sensing off the grid result (the completion problem) \cite{Tang:2013fo} building upon the dual polynomial constructed for the super-resolution problem (the full data case) \cite{Candes:2014br}. It shows that the number of random linear measurements required for exact recovery approaches the information theoretical limit. In addition to exact recovery from noise-free measurements, the {\em dual certificate} for the full data case can also be utilized to derive near-minimax denoising performance \cite{Bhaskar:2013ki,Tang:2013gd}, approximate support recovery \cite{FernandezGranda:2013wp,li2018approximate}, and robust recovery from observations corrupted by outliers \cite{Tang:2014jk,fernandez2016demixing}. We expect that the dual polynomial constructed for tensor decomposition will play a similar role for tensor inverse problems, enabling the development of tensor results  parallel to their matrix counterparts such as matrix completion, denoising, and robust principal component analysis. We leave these as our future work.

\section{Proof of Theorem \ref{thm:main}}\label{sec:proof}

\subsection{Proof Outline}

\noindent
\makebox{
\begin{minipage}{0.55\textwidth}
The proof of Theorem \ref{thm:main} relies on the construction of a dual polynomial that satisfies the Boundedness and Interpolation  Property \eqref{eqn:bip}.
 Towards that end, we first partition $\bbK$ into the far region (controlled by Lemma \ref{lem:far:region}) and the  near region. To control the dual polynomial in the near region, we use an angular parametrization to 
further divide it into near vertex region 
(controlled by Lemma \ref{lem:near:vertex:region}) and  near band region (controlled by Lemma \ref{lem:near:band:region}).  
In the end, we can show the constructed dual polynomial satisfies the \textsf{BIP} in the whole region.  We summarize the proof map on the right. 
\end{minipage}
~~~~
\begin{minipage}{0.43\textwidth}
\vspace{-.2cm}
\includegraphics[width =1\textwidth]{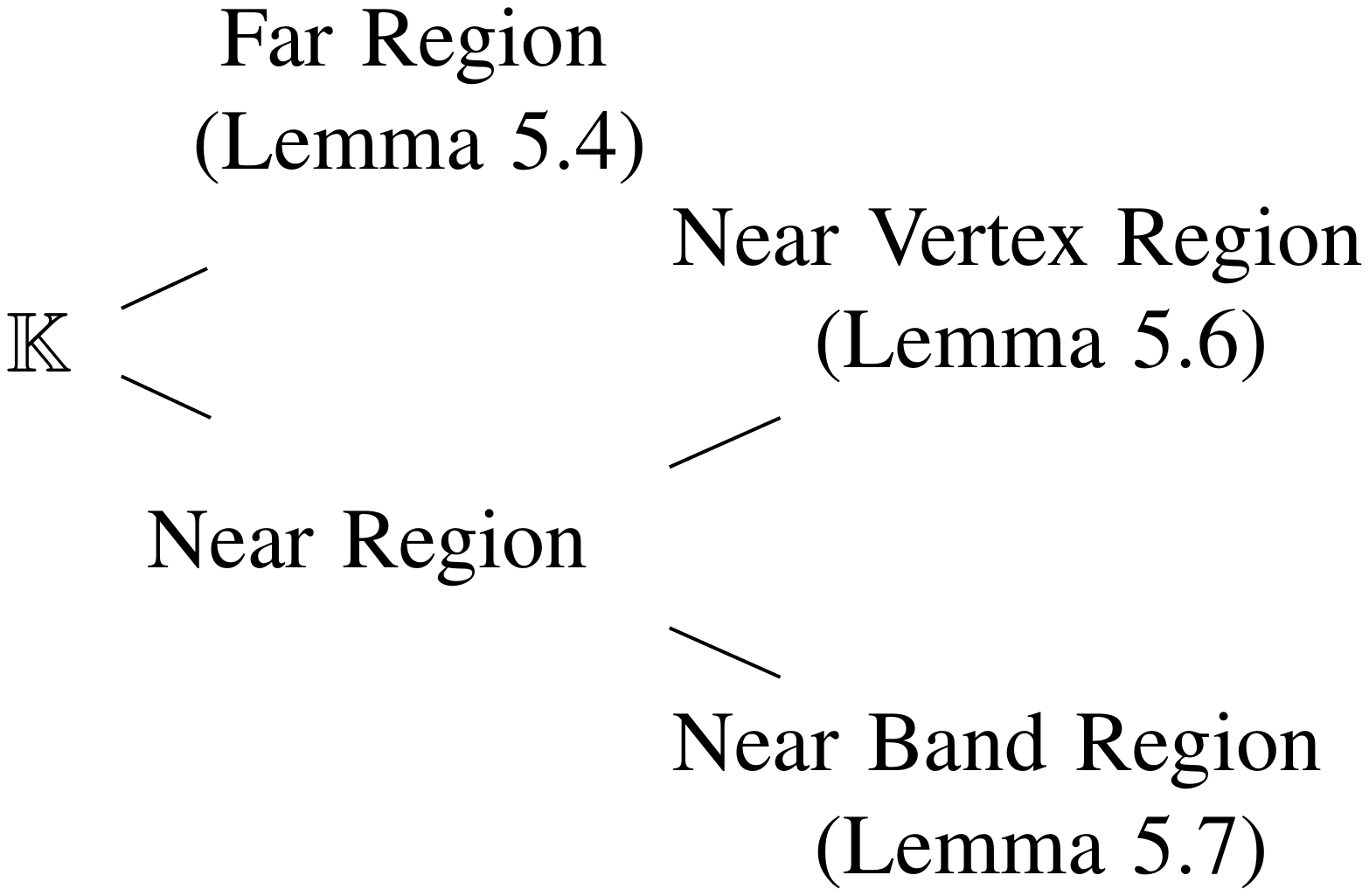}
\end{minipage}
}

\subsection{Minimal Energy Construction}
Since the \textsf{BIP} \eqref{eqn:bip} (especially the Boundedness property \eqref{eqn:bip:b}) is hard to enforce directly, we start from  a candidate {\em dual certificate} or pre-certificate $\tQ$ in the subdifferntial set $\partial \|\tT\|_*$
defined by \eqref{eqn:subgradient:condition}:
\begin{align*}	
\langle \tQ, \u_{p}^{\star}\otimes\v_{p}^{\star}\otimes\w_{p}^{\star}\rangle & =   1,\ \tmop{for} p\in[r]\\
\langle \tQ, \u\otimes \v\otimes \w \rangle & \leq  1 {} ,\ \tmop{for}   ( \u,\v,\w ) \in \bbK
\end{align*}	
which essentially characterizes the optimal solution set of following optimization
\begin{align}
\label{eqn:mip:optmization}
\maximize_{(\u,\v,\w)\in\bbK}\	\langle \tQ, \u\otimes \v\otimes \w \rangle
\end{align}
Then applying the Karush-Kuhn-Tucker (KKT)  conditions to the constrained optimization \eqref{eqn:mip:optmization}, we can  further relax the subdifferential conditions \eqref{eqn:subgradient:condition} to a set of linear constraints.
\begin{lemma}\label{lem:relaxedcondition}
The following conditions are necessary for \eqref{eqn:subgradient:condition}:
\begin{equation}
\begin{aligned}
\sum_{j,k} Q_{i j k} \v_{p}^{\star} ( j ) \w_{p}^{\star} ( k )   &=
\u_{p}^{\star} ( i ) ,\forall i\in[n],\forall p\in[r] ; \\
\sum_{i,k} Q_{i j k} \u_{p}^{\star} ( i ) \w_{p}^{\star} ( k )   &=
\v_{p}^{\star} ( j ) ,\forall i\in[n],\forall p\in[r] ;\\
\sum_{i,j} Q_{i j k} \u_{p}^{\star} ( i ) \v_{p}^{\star} ( j )   &=
\w_{p}^{\star} ( k ) ,\forall i\in[n],\forall p\in[r]
\end{aligned}
\quad \stackrel{or}{\iff} \quad
\begin{aligned}
& \tQ{\times}_2\v_p^\star
{\times}_3\w_p^\star  = \u_p^\star, \forall p\in[r]; \\
& \tQ{\times}_1\u_p^\star
{\times}_3\w_p^\star  = \v_p^\star, \forall p\in[r]; \\
& \tQ{\times}_1\u_p^\star
{\times}_2\v_p^\star   = \w_p^\star, \forall p\in[r]
\end{aligned}
\label{eqn:minimum:energy}
\end{equation}
where $\{{\times}_k\}$  are the $k$-mode tensor-vector product \cite{kolda2009tensor} whose definitions are apparent from context.
\end{lemma}
The proof of Lemma \ref{lem:relaxedcondition} is given in Appendix \ref{App:proof:lem:relaxedcondition}.

Apparently, the subdifferential conditions \eqref{eqn:subgradient:condition} is necessary for the \textsf{BIP} \eqref{eqn:bip},
but generally not sufficient, by comparing the second line of \eqref{eqn:subgradient:condition}  and the Boundedness Property \eqref{eqn:bip:b}. Indeed, as we argued before, any $\tQ$ satisfying the \textsf{BIP} is an interior point of the subdifferential $\partial \|\cdot\|_*(\tT)$. To satisfy the Boundedness Property \eqref{eqn:bip:b},   we further
minimize the energy $\|\tQ\|^{2}_{F} = \sum_{i j k} Q_{i j k}^{2}$ in the
hope that this will push $q ( \u,\v,\w )$ towards zero such that $\tQ$ is an interior point of  $\partial \|\cdot\|_*(\tT)$.
Thus, we propose solving the following   \emph{minimum-energy}
problem to obtain a pre-certificate:
\begin{align}
\operatorname*{minimize}_Q   \ \frac{1}{2}\|\tQ\|_{F}^{2} \
\st \eqref{eqn:minimum:energy}  \label{eqn:minimum:norm}
\end{align}
\begin{lemma}[Explicit form of the  pre-certificate] \label{lem:normalequation}
The solution of the least-norm problem (\ref{eqn:minimum:norm}) has the form (normal equation)
\begin{align}\label{eqn:Q:explicit:form}
\tQ= \sum_{p=1}^r ( \valpha^\star_{p} \otimes \v_{p}^{\star} \otimes
\w_{p}^{\star} +\u_{p}^{\star} \otimes \vbeta^\star_{p} \otimes \w_{p}^{\star}
+\u_{p}^{\star} \otimes \v_{p}^{\star} \otimes \vgamma^\star_{p} )
\end{align}
with the unknown coefficients $\{ \valpha^\star_{p} , \vbeta^\star_{p} , \vgamma^\star_{p} \}_{p=1}^r$ being chosen such that $\tQ$ in \eqref{eqn:Q:explicit:form} satisfies \eqref{eqn:minimum:energy}. So we get an explicit form of a pre-certificate
\begin{align}
	q ( \u,\v,\w ) & = \langle \tQ,\u \otimes \v \otimes \w \rangle\nonumber                                             \\
	               & = \sum_{p=1}^r [ \langle \valpha^\star_{p} ,\u \rangle \langle \v_{p}^{\star} ,\v
\rangle \langle \w_{p}^{\star} ,\w \rangle
+ \langle \u_{p}^{\star} ,\u \rangle
\langle \vbeta^\star_{p} ,\v \rangle \langle \w_{p}^{\star} ,\w \rangle
+ \langle
\u_{p}^{\star} ,\u \rangle \langle \v_{p}^{\star} ,\v \rangle \langle \vgamma^\star_{p}
,\w \rangle].\label{eqn:dual:polynomial:explicit:form}
\end{align}
\end{lemma}
The proof of Lemma \ref{lem:normalequation} is given in Appendix \ref{App:proof:lem:normalequation}.

To obtain some intuition of what these dual-polynomial coefficients $\{\valpha^\star_{p} , \vbeta^\star_{p} , \vgamma^\star_{p}\}_{p=1}^r$ would look like, let us
assume $\{\u_{p}^{\star}\}_{p=1}^r$, $\{\v_{p}^{\star}\}_{p=1}^r$, $\{\w_{p}^{\star}\}_{p=1}^r$ are almost orthogonal and  plug the explicit form  of $\tQ$ \eqref{eqn:Q:explicit:form} into the first equation in \eqref{eqn:minimum:energy}
\begin{align}\label{eqn:estimate:1}
\valpha_p^\star+\u_p^\star\lg\vbeta_p^\star,\v_p^\star\rg+\u_p^\star\lg\vgamma_p^\star,\w_p^\star\rg\approx \u_p^\star.
\end{align}
Then multiplying $\u_p^{\star\top}$  on both sides gives
\begin{align}\label{eqn:estimate:2}
\lg\valpha_p^\star,\u_p^\star\rg+\lg\vbeta_p^\star,\v_p^\star\rg+\lg\vgamma_p^\star,\w_p^\star\rg\approx 1.
\end{align}
Finally combining \eqref{eqn:estimate:1} and \eqref{eqn:estimate:2} together with the symmetry property of \eqref{eqn:Q:explicit:form}, we get  these coefficients  $\{ \valpha^\star_{p} , \vbeta^\star_{p} , \vgamma^\star_{p} \}_{p=1}^r$ are located approximately at  $
\{ \u^\star_{p}/3 , \v^\star_{p}/3 , \w^\star_{p}/3 \}_{p=1}^r$. The accurate description of this phenomenon is given by the following lemma with the proof listed in Appendix \ref{App:proof:lem:Bound_coef}.
\begin{lemma}[ Control the dual polynomial coefficients]\label{lem:Bound_coef}
Under Assumptions II and III together with $r=o({n^2}/{\kappa(\log n)^2})$, the following estimates are valid for sufficiently large $n$:
\begin{align*}
	\left\|\mA - \frac{1}{3}\mU\right\| & \leq 2\kappa(\log n)\left(\frac{\sqrt r}{n}+c{\frac{r}{n^{1.5}}}\right);  \\ 
	\left\|\mB - \frac{1}{3}\mV\right\| & \leq 2\kappa(\log n)\left(\frac{\sqrt r}{n}+c{\frac{r}{n^{1.5}}}\right) ; \\
	\left\|\mC - \frac{1}{3}\mW\right\| & \leq 2\kappa(\log n)\left(\frac{\sqrt r}{n}+c{\frac{r}{n^{1.5}}}\right)
\end{align*}
where the norm $\|\cdot\|$ is the matrix spectral norm and 
\[
\mA=\begin{bmatrix}\valpha_1^\star, \cdots, \valpha_r^\star\end{bmatrix},
\mB=\begin{bmatrix}\vbeta_1^\star, \cdots, \vbeta_r^\star\end{bmatrix},
\mC=\begin{bmatrix}\vgamma_1^\star, \cdots, \vgamma_r^\star\end{bmatrix}, 
\mU= \begin{bmatrix}\u_{1}^{\star}, \cdots, \u_{r}^{\star}\end{bmatrix},
\mV= \begin{bmatrix}\v_{1}^{\star}, \cdots, \v_{r}^{\star}\end{bmatrix},
\mW= \begin{bmatrix}\w_{1}^{\star}, \cdots, \w_{r}^{\star}\end{bmatrix}.
\]
\end{lemma}

\subsection{Far Region}
For a parameter $\delta\in(0,1)$, the far region  is defined by
\begin{align}
& \mathcal{F}(\delta) := \bigcap_{p=1}^r \{( \u,\v,\w ) \in \bbK:
| \langle \u,\u_{p}^{\star}
\rangle | \leq \delta  \text{ or } | \langle \v,\v_{p}^{\star} \rangle | \leq \delta  \text{ or }
|\langle \w,\w_{p}^{\star} \rangle | \leq \delta  \},\label{def:far:region}
\end{align}
which consists of points $(\u,\v,\w)$ in $\bbK$ that are far away (in the angular sense) from
\begin{align}
\bbS^\star:=\{(\pm\u^\star_p, \pm\v^\star_p, \pm\w^\star_p): p = 1, \ldots, r\}
\end{align}
in at least one coordinate of $(\u,\v,\w)$. For $n=3$ and $r = 2$, the far region projected onto the unit sphere $\{\u: \|\u\|_2 = 1\}$ is shown in Fig.  \ref{fig:far-away}.

\noindent
\makebox{
\begin{minipage}{0.5\textwidth}
\begin{figure}[H]
\centering
\includegraphics[width =0.6\textwidth]{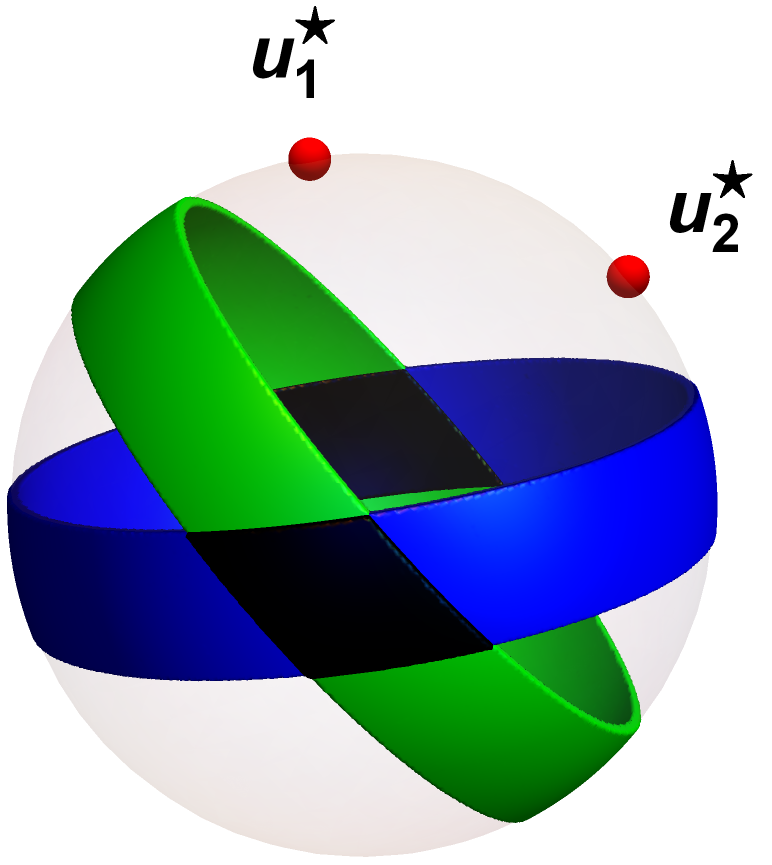}
\caption{Projection of the far region in the $\u$ coordinate. The blue band represents the region $\{\u: |\langle \u, \u_1^\star\rangle| \leq \delta\}$ that is far away from $\u_1^\star$, while the green region $\{\u: |\langle \u, \u_2^\star\rangle| \leq \delta\}$ is the far-region associated with $\u_2^\star$. The far region is their intersection $\bigcap_{p=1}^2\{\u: |\langle \u, \u_p^\star\rangle| \leq \delta\}$,  consisting of the two black diamonds.}
\label{fig:far-away}
\end{figure}	
\end{minipage}
~~~
\begin{minipage}{0.5\textwidth}
\begin{figure}[H]
\centering
\includegraphics[width =0.55\textwidth]{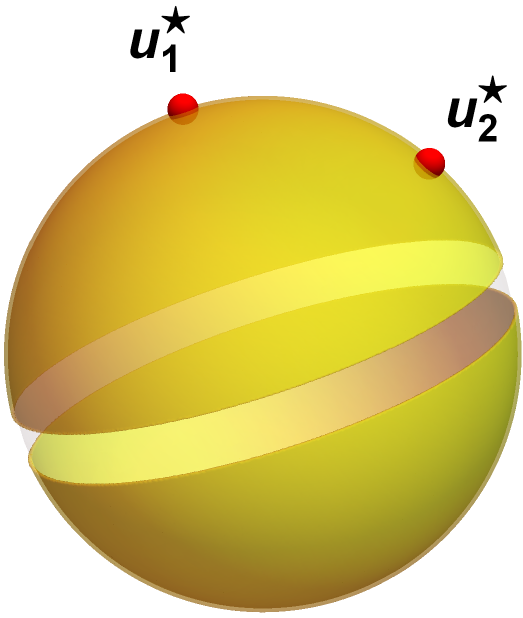}
\caption{The two yellow spherical caps form the near region $\mathcal{N}_1(\delta)$ around the point $(\u_1^\star, \v_1^\star, \w_1^\star)$ projected onto the $\u$ coordinates. $\calN_2(\delta)$, which is not shown here, consists of another two spherical caps. The union of $\calN_1(\delta), \calN_2(\delta)$ and the far region $\mathcal{F}(d)$ shown in Fig.  \ref{fig:far-away} will cover the entire sphere $\{\u: \|\u\|=1\}$.}
\label{fig:near_region}
\end{figure}
\end{minipage}}

\subsubsection{Controlling in Far Region}
Instead of bounding   the dual polynomial $q$ directly, we will bound its absolute value  $|q|$. To obtain some intuition of how to bound it,
we rewrite the explicit form \eqref{eqn:dual:polynomial:explicit:form} as follows
\begin{align}
&q ( \u,\v,\w )
\nn
\\
&= \sum_{p=1}^r \bigg[ \langle \valpha^\star_{p}-\frac{1}{3}\u_p^\star ,\u \rangle \langle \v_{p}^{\star} ,\v
\rangle \langle \w_{p}^{\star} ,\w \rangle
+ \langle \u_{p}^{\star} ,\u \rangle
\langle \vbeta^\star_{p}-\frac{1}{3}\v_p^\star,\v \rangle \langle \w_{p}^{\star} ,\w \rangle
+ \langle
\u_{p}^{\star} ,\u \rangle \langle \v_{p}^{\star} ,\v \rangle \langle \vgamma^\star_{p} -\frac{1}{3}\w_p^\star  ,\w \rangle \bigg]\label{eqn:far:terms:0}\\
&~~~~
+\sum_{p=1}^r  \langle\u_{p}^{\star} ,\u \rangle \langle \v_{p}^{\star} ,\v \rangle \langle\w_p^\star  ,\w \rangle \label{eqn:far:terms:2}.
\end{align}
The main idea is first using the closeness of    $\{ \valpha^\star_{p} , \vbeta^\star_{p} , \vgamma^\star_{p} \}_{p=1}^r$ and $
\{ \u^\star_{p}/3 , \v^\star_{p}/3 , \w^\star_{p}/3 \}_{p=1}^r$  to bound  \eqref{eqn:far:terms:0} and then using angular-distance between $\calF(\delta)$ and  $(\u_p^\star,\v_p^\star,\w_p^\star),\forall p$ to bound  \eqref{eqn:far:terms:2}.

The accurate argument is made by the following lemma with the proof given in  Appendix \ref{App:proof:lem:Far_away}.
\begin{lemma}[Controlling in Far Region]\label{lem:far:region}
Under  Assumptions I,  II, III, if $r\ll n^{1.25}$ and $r\leq \frac{n}{24\delta c^2}$ for $\delta\in(0,\frac{1}{24}]$, then for sufficiently large $n$, we have    $|q(\u,\v,\w)|<1$  in $\calF(\delta).$
\end{lemma}

\subsection{Near Region}


For the union of the far and near regions to cover the entire region $\bbK$,  	we define the near region as
\begin{align}	
\calN(\delta):=&\bbK\setminus \calF(\delta)
=\bigcup_{p=1}^r\{(\u,\v,\w)\in\bbK:|\langle\u_p^\star,\u\rangle|\geq \delta, |\langle\v_p^\star,\v\rangle|\geq \delta,
|\langle\w_p^\star,\w\rangle|\geq \delta\}:=\bigcup_{p=1}^r \calN_p(\delta)
\label{def:near:region}
\end{align}
with each individual near region  $\mathcal{N}_p(\delta)$
close to $(\u_p^\star,\v_p^\star,\w_p^\star)$ in all coordinate of $(\u,\v,\w)$.

For $n=3$, $r=2$,
we plot the near region $\mathcal{N}_1(\delta)$ projected onto the sphere $\{\u: \|\u\|_2 = 1\}$ in Fig.  \ref{fig:near_region}.

\subsubsection{Angular Parametrization of Near Region}

In order to show the dual polynomial satisfying the \textsf{BIP} in the entire near region $\calN(\delta)$, we use the ``Divide-and-conquer" idea to bound the dual polynomial in each individual near region $\calN_p(\delta)$ for $p\in[r].$ The main technique used to control each individual near region is applying angular parametrization to each individual near region.

As the domain $\bbK$ is essentially a direct product of spheres, we  re-parameterize each individual near region $\calN_p(\delta)$ in the angular sense.
Without loss of generality,  let us consider $p=1$. Pick $(\x,\y,\z)\in\bbK$ such that $\x \perp \u_1^\star, \y \perp \v_1^\star, \z \perp \w_1^\star$ and consider the parameterized points 
\begin{equation} \label{eqn:angular:parameter}
(\u(\theta_1),\v(\theta_2),\w(\theta_3))\in\bbK\quad\tmop{with}\quad
\begin{cases}
\u(\theta_1) = \u_1^\star \cos(\theta_1)\   + \x\sin(\theta_1)  \\
\v(\theta_2) = \v_1^\star \cos(\theta_2)\   + \y\sin(\theta_2)  \\
\w(\theta_3) = \w_1^\star \cos(\theta_3)   +  \z\sin(\theta_3).
\end{cases}
\end{equation}

\noindent
\makebox{
\begin{minipage}{0.6\textwidth}
When $\theta_1$ ranges from $0$ to $\pi$, $\u(\theta_1)$ traces out a 2D semi-circle that starts at $\u_1^\star$, passes through $\x$, and finally reaches $-\u_1^\star$; while for a fixed $\theta_1\in[0,\pi]$, the set $\bigcup_{\x \perp \u_1^\star}\{\u(\theta_1)\}$ parameterizes all the points on $\S^{n-1}$ having an angle of $\theta_1$ with $\u^\star_1$. The same properties hold for $\v(\theta_2)$ and $\w(\theta_3)$.	
This parametrization projected onto the $\u$ coordinate is shown on right.
\end{minipage}
 ~
\begin{minipage}{0.4\textwidth}
\centering
\includegraphics[width =0.7\textwidth]{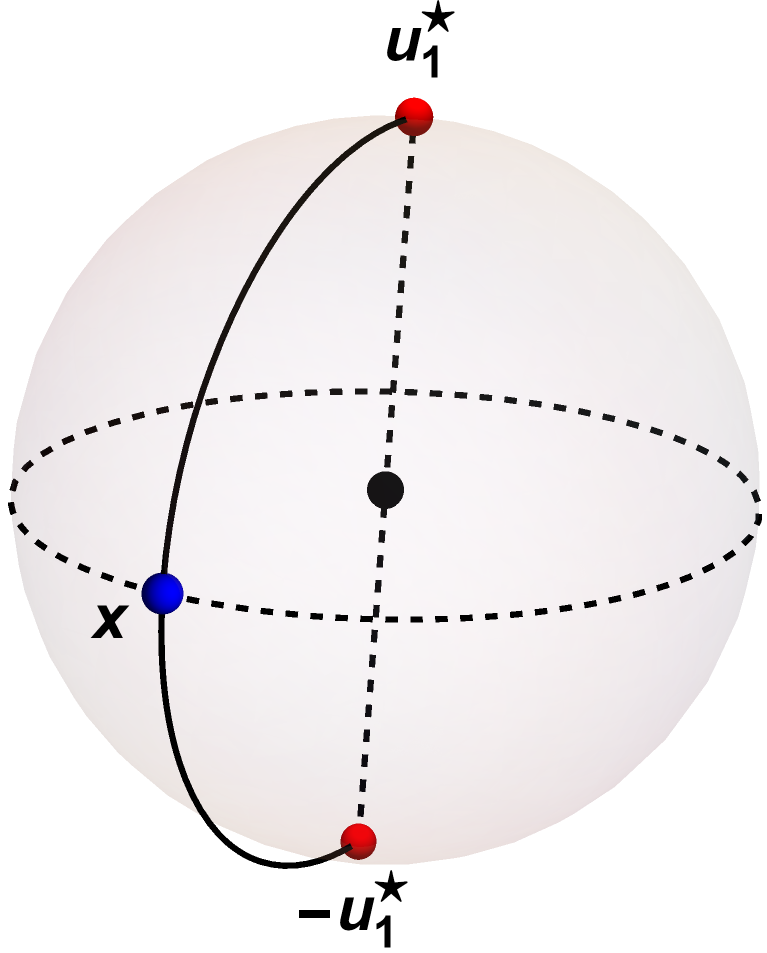}
\end{minipage}}

In fact, using this angular parametrization, the individual
near region $\calN_1(\delta)$ in \eqref{def:near:region} can be expressed as
\begin{align}
\calN_1(\delta)
=&\bigcup_{(\x,\y,
\z): \x \perp \u_1^\star, \y \perp \v_1^\star, \z \perp \w_1^\star} \{(\u(\theta_1),\v(\theta_2),\w(\theta_3)):  |\cos(\theta_i)|\geq\delta, \theta_i\in[0,\pi],   i=1,2,3\}.
\label{eqn:param:K}
\end{align}
\begin{proposition}[Near Angular Region]\label{pro:superset}
For any $\delta\in(0,1)$, the near region $\calN_1(\delta)$ is contained in the following set
\begin{align}
\calN_1(\delta)
\subset&
\bigcup_{(\x,\y,
	\z): \x \perp \u_1^\star, \y \perp \v_1^\star, \z \perp \w_1^\star} \{(\u(\theta_1),\v(\theta_2),\w(\theta_3)):(\theta_1,\theta_2,\theta_3)\in\bbN(\delta)\}
\label{eqn:larger:set}
\end{align}
with the near angular region $\bbN(\delta)$ defined by
\begin{align}\label{eqn:angular:near}
\bbN(\delta):=&\left\{(\theta_1,\theta_2,\theta_3): \theta_i\in\left[0,\frac{\pi}{2}-\delta\right]\cup\left[\frac{\pi}{2}+\delta,\pi\right], i=1,2,3 \right\}.
\end{align}
\end{proposition}
\begin{proof}
Since the function $|\cos (\theta)|$ is symmetric at $\frac{\pi}{2}$ on the interval $[0, \pi]$ and is decreasing on $[0, \pi/2]$, we know that
	$\{\theta:|\cos(\theta)|\geq\delta\}\cap[0,\pi] = [0, \arccos(\delta)]\cup [\pi-\arccos(\delta),\pi]$. Note that $\arccos (\delta)=\frac{\pi}{2}-\arcsin (\delta)$ and $\delta < \arcsin (\delta)$, so we get $\{\theta:|\cos(\theta)|\geq\delta\}\cap[0,\pi] \subset [0,\frac{\pi}{2}-\delta]\cup[\frac{\pi}{2}+\delta,\pi]$. The inclusion \eqref{eqn:larger:set} follows  from \eqref{eqn:param:K} immediately.
\end{proof}

The near angular region $\bbN(\delta)$ contains the eight cubes with side length $\frac{\pi}{2}-\delta$, located at the eight corners of the cube $[0,\pi]\times[0,\pi]\times[0,\pi]$. Moreover, one can see that the smaller the parameter $\delta$ is, the larger the near angular region $\bbN(\delta)$ will be. In particular, when $\delta$   approaches to zero, the near angular region $\bbN(\delta)$ becomes the whole   cube $\bbN(0)=[0,\pi]\times[0,\pi]\times[0,\pi]$. The near angular region $\bbN(\delta)$ is plotted in Fig.  \ref{fig:angular_region}.

\noindent
\makebox{
\begin{minipage}{0.5\textwidth}
\begin{figure}[H]
\centering
\includegraphics[width =0.7\textwidth]{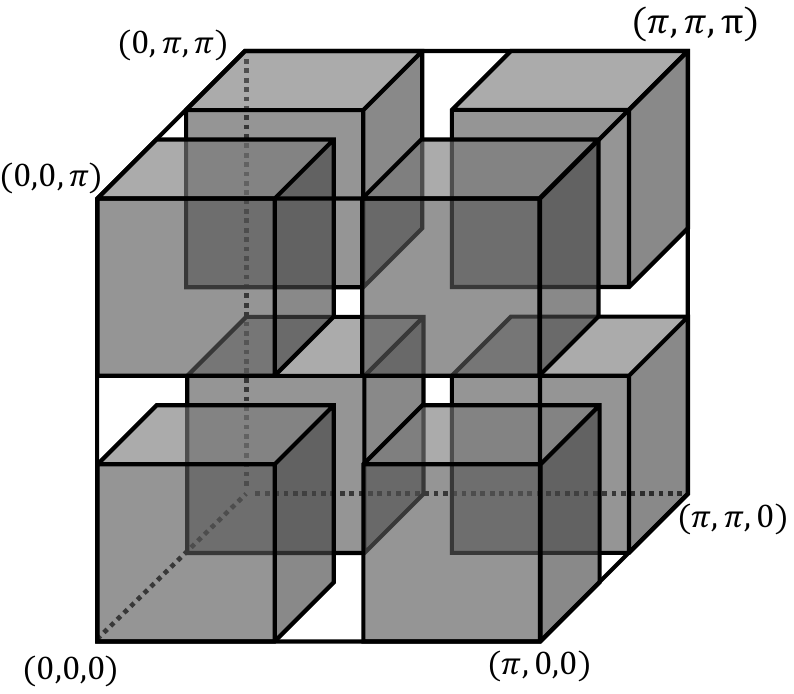}
\caption{The eight gray cubes of side-length $\pi/2-\delta$ at the corners form the near angular region $\bbN(\delta)$.}
\label{fig:angular_region}
\end{figure}	
\end{minipage}
~~~
\begin{minipage}{0.5\textwidth}
\begin{figure}[H]
\centering
\includegraphics[width =0.7\textwidth]{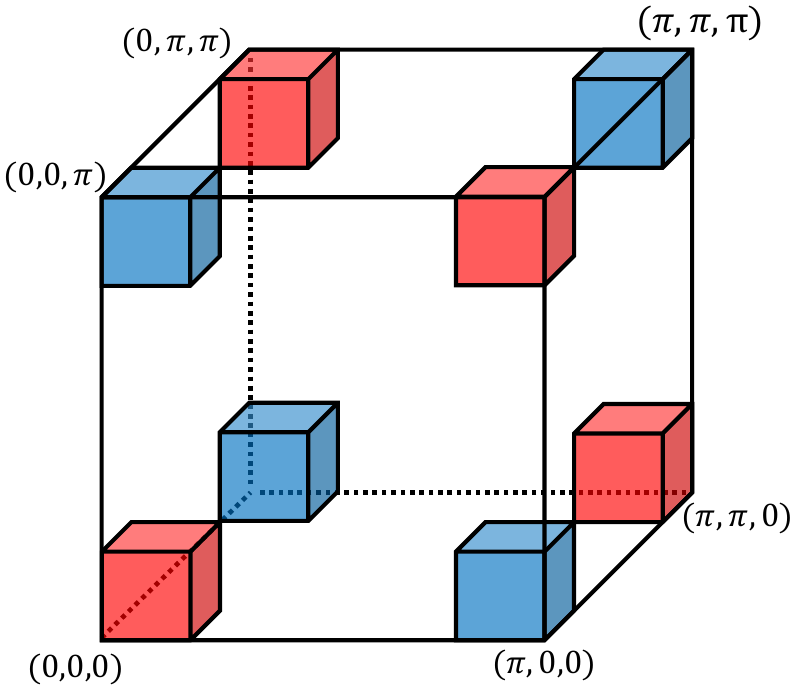}
\caption{The eight colored cubes of size $\delta_v\times\delta_v\times\delta_v$ form the near vertex region
	$\bbN_v(\delta_v)$: the red ones are corresponding to the vertexes in $\bbS^\star$ while the blue ones are corresponding to other vertexes in the  cube.}
\label{fig:vertex_region}
\end{figure}
\end{minipage}}

\subsubsection{Angular Parametrization of Dual Polynomial}
Evaluating the dual polynomial $q(\u,\v,\w)$ at $(\u(\theta_1),\v(\theta_2),\w(\theta_3))$ in \eqref{eqn:angular:parameter}, we get the angular dual polynomial $F(\theta_1,\theta_2,\theta_3)
:=
q(\u(\theta_1),\v(\theta_2),\w(\theta_3))$ as
\begin{align}
F(\theta_1,\theta_2,\theta_3)	
=& q(\u_1^\star,\v_1^\star,\w_1^\star)\cos(\theta_1)\cos(\theta_2)\cos(\theta_3)
+
q(\u_1^\star,\v_1^\star,\z)\cos(\theta_1)\cos(\theta_2)\sin(\theta_3)\nn\\
&+q(\u_1^\star,\y,\w_1^\star)\cos(\theta_1)\sin(\theta_2)\cos(\theta_3)
+q(\x,\v_1^\star,\w_1^\star)\sin(\theta_1)\cos(\theta_2)\cos(\theta_3)\nn\\
&+
q(\u_1^\star,\y,\z)\cos(\theta_1)\sin(\theta_2)\sin(\theta_3)
+q(\x,\v_1^\star,\z)\sin(\theta_1)\cos(\theta_2)\sin(\theta_3)
\nn\\
&+q(\x,\y,\w_1^\star)\sin(\theta_1)\sin(\theta_2)\cos(\theta_3)+q(\x,\y,\z)\sin(\theta_1)\sin(\theta_2)\sin(\theta_3).
\label{eqn:angular:form}
\end{align}
Among these $8$ terms, the first term is $\cos(\theta_1)\cos(\theta_2)\cos(\theta_3)$ since $q(\u_1^\star,\v_1^\star,\w_1^\star) = 1$.
The next three terms involving one sine function are zero as, for example,
\begin{align*}
q(\u_1^\star,\v_1^\star,\z)
&= \tQ{\times}_1\u_1^\star{\times}_2\v_1^\star{\times}_3\z=\w_1^\star{\times}_3\z={\w_1^\star}^\top \z=0,
\end{align*}
where we have used $Q{\times}_1 \u_1^\star {\times}_2 \v_1^\star = \w_1^\star$ and the third equality of \eqref{eqn:minimum:energy}.
Hence, we get a more concise form of $F$:
\begin{align}
F(\theta_1,\theta_2,\theta_3)
=& \cos(\theta_1)\cos(\theta_2)\cos(\theta_3)
+q(\u_1^\star,\y,\z)\cos(\theta_1)\sin(\theta_2)\sin(\theta_3) +q(\x,\v_1^\star,\z)\sin(\theta_1)\cos(\theta_2)\sin(\theta_3) \nn\\
&+q(\x,\y,\w_1^\star)\sin(\theta_1)\sin(\theta_2)\cos(\theta_3) +q(\x,\y,\z)\sin(\theta_1)\sin(\theta_2)\sin(\theta_3).
\end{align}
By further bounding the other quantities $q(\u_1^\star,\y,\z)$, $q(\x,\v_1^\star,\z)$, $q(\x,\y,\w_1^\star)$ and $q(\x,\y,\z)$, we get the following lemma to uniformly upper-bound   $F(\theta_1,\theta_2,\theta_3)$ with the proof given in Appendix \ref{App:proof:lem:Parameterize}.
\begin{lemma}[Upper Bound of Angular Dual Polynomial]\label{lem:Parameterize}
Under  Assumptions I, II, III, if $r\leq n^{1.25-1.5 r_c}$ with $r_c\in(0,\frac{1}{6})$,
then for sufficiently large $n$, we have
\begin{align}
&|F(\theta_1,\theta_2,\theta_3)|
\leq |\cos(\theta_1)\cos(\theta_2)\cos(\theta_3)| + |\sin(\theta_1)\sin(\theta_2)\sin(\theta_3)|+\frac{4}{3}\tau(\log n) n^{-r_c}.
\label{eqn:bound:angular:form}
\end{align}
\end{lemma}

 \subsubsection{Angular Parametrization of  Boundedness and Interpolation Property}
By Proposition \ref{pro:superset}, a sufficient condition for the \textsf{BIP} \eqref{eqn:bip} to hold in the individual  near region $\calN_1(\delta)$, is the following Angular Boundedness and Interpolation Property  (\textsf{Angular-BIP}):
\begin{subequations}\label{eqn:angular:bip}
	\begin{align}
	F(\theta_1,\theta_2,\theta_3)&=1 \text{ in }
	\bbS^\star\label{eqn:angular:interpolation}
	\ \
	(\textsf{Angular Interpolation})\\
	F(\theta_1,\theta_2,\theta_3)&<1 \text{ in }
	\bbN(\delta)\setminus \bbS^\star
	\ \ (\textsf{Angular Boundedness})
	\label{eqn:angular:boundedness}
	\end{align}
\end{subequations}
with $\bbS^\star:=\{(0,0,0),(0,\pi,\pi),(\pi,0,\pi),(\pi,\pi,0)\}$ such that $\u(\theta_1)\otimes\v(\theta_2)\otimes\w(\theta_3)=\u_1^\star\otimes\v_1^\star\otimes\w_1^\star$ for any $(\theta_1,\theta_2,\theta_3)\in\bbS^\star.$

Similar as before,  the Angular Interpolation property \eqref{eqn:angular:interpolation} is a consequence of the construction process. In the rest of the paper, we will focus on showing the Angular Boundedness property \eqref{eqn:angular:boundedness}. Specifically, we will divide the near angular region into near vertex region and near band region, and then control the 
angular dual polynomial $F$ in both near vertex region and near band region.

\subsubsection{Near Vertex Region}
The near vertex region, denoted by $\bbN_v(\delta_v)$, is defined as the union of the eight small cubes all with side length $\delta_v$ in 8 corners of the  cube $[0,\pi]^3$. We plot the near vertex region $\bbN_v(\delta_v)$ in Fig.  \ref{fig:vertex_region}. Comparing with the definition of the near angular region $\bbN(\cdot)$, the near vertex region is also an near angular region but with a different parameter:
\begin{align}
\bbN_v(\delta_v)=\bbN(\frac{\pi}{2}-\delta_v).
\label{eqn:vertex:vs:near}
\end{align}
Without loss of generality, we can always assume the near vertex region $\bbN_v(\delta_v)$ is included in the near angular region $\bbN(\delta)$; otherwise, we only need to show the \textsf{Angular-BIP} holds in $\bbN_v(\delta_v)$. This assumption together with \eqref{eqn:vertex:vs:near}  implies
\begin{align}
\label{assumption:delta}
\delta_v\leq \frac{\pi}{2}-\delta.
\end{align}
Note that ${\pi}/{2}-\delta$ is the side length of the corner-cubes in $\bbN(\delta).$

\paragraph{Controlling in Near Vertex Region}	
To control the angular dual polynomial $F$ in the near vertex region $\bbN_v(\delta_v)$, we further classify the eight small cubes in $\bbN_v(\delta_v)$ into two groups depending on if their vertices are in $\bbS^\star$ or not.

\begin{lemma}[Controlling in Near Vertex Region]\label{lem:near:vertex:region}
Under  Assumptions I, II, III, if $r\ll n^{1.25}$, then for any $\xi_i\in\bigg(-\frac{\sqrt 2-1}{3},\frac{\sqrt 2-1}{3}\bigg)$, we have
\begin{align}
&F(\theta_1+\xi_1,\theta_2+\xi_2,\theta_3+\xi_3) \leq 1
\label{eqn:near:vertex:region}
\end{align}
for $(\theta_1,\theta_2,\theta_3)\in\{(0,0,0),(0,\pi,\pi),(\pi,0,\pi),(\pi,\pi,0)\}$ and
\begin{align}
&F(\theta_1+\xi_1,\theta_2+\xi_2,\theta_3+\xi_3) <0
\label{eqn:near:vertex:region:2}
\end{align}
for $(\theta_1,\theta_2,\theta_3)\in\{(\pi,\pi,\pi),(\pi,0,0),(0,\pi,0),(0,0,\pi)\}$.  Here, equality in \eqref{eqn:near:vertex:region} holds only if  $\xi_1=\xi_2=\xi_3=0$.
\end{lemma}
The proof of Lemma \ref{lem:near:vertex:region} is  in Appendix \ref{App:proof:lem:near:region}.

\begin{remark}Lemma \ref{lem:near:vertex:region} proves the \textsf{Angular-BIP} holds  in the near vertex region $\bbN_v(\delta_v)$ with
$\delta_v=\frac{\sqrt2-1}{3}$:
\begin{align}
F(\theta_1,\theta_2,\theta_3)&=1 \text{ in }
\bbS^\star\nn \\
F(\theta_1,\theta_2,\theta_3)&<1 \text{ in }
\bbN_v(\delta_v)\setminus \bbS^\star
\nn
\end{align}
\end{remark}

\subsubsection{Near Band Region}
The near band region is introduced to cover the remaining region $\bbN(\delta)\setminus \bbN_v(\delta_v)$.
Invoking the definitions of the near angular region \eqref{eqn:angular:near} and the near vertex region \eqref{eqn:vertex:vs:near}:	
\begin{align*}
\bbN(\delta)=&\left\{(\theta_1,\theta_2,\theta_3): \theta_i\in\left[0,\frac{\pi}{2}-\delta\right]\cup\left[\frac{\pi}{2}+\delta,\pi\right]\right\}
\\
\bbN_v(\delta_v)=&\left\{(\theta_1,\theta_2,\theta_3): \theta_i\in[0,\delta_v]\cup[\pi-\delta_v,\pi]\right\}
\end{align*}
we have
\begin{align}
\bbN(\delta)\setminus \bbN_v(\delta_v)
&=\left\{(\theta_1,\theta_2,\theta_3): \theta_i\in\left(\delta_v,\frac{\pi}{2}-\delta\right)\cup\left(\frac{\pi}{2}+\delta,\pi-\delta_v\right)
\right\}\cap \bbN(\delta),
\label{def:remain:region}
\end{align}
which is nonempty since $\delta_v\leq \pi/2-\delta$ by the assumption \eqref{assumption:delta}.
We plot the remaining region $\bbN(\delta)\setminus \bbN_v(\delta_v)$ projected onto the $(\theta_1,\theta_2)$-coordinates   in Fig.  \ref{fig:remain_region}.

\noindent
\makebox{
\begin{minipage}{0.5\textwidth}
\begin{figure}[H]
\centering
\includegraphics[width =0.6\textwidth]{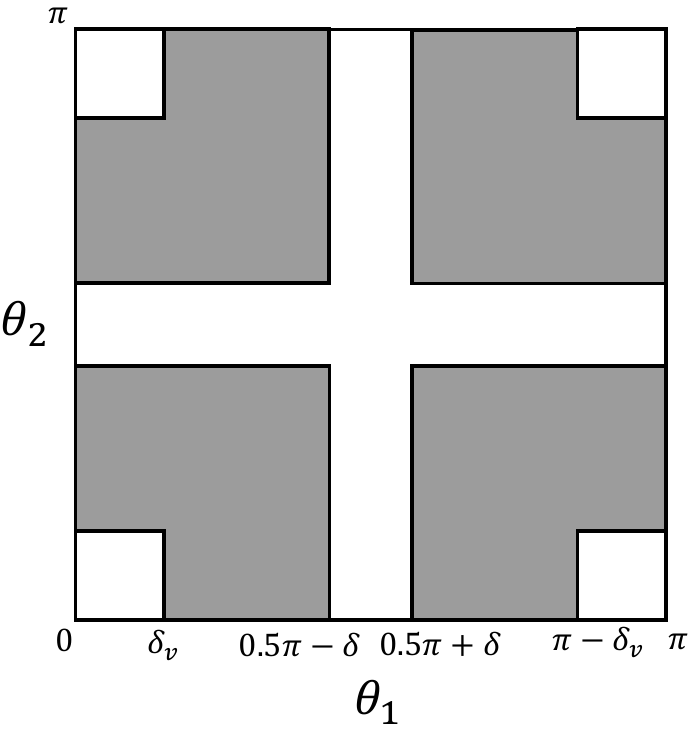}
\caption{The
	remaining region $\bbN(\delta)\setminus \bbN_v(\delta_v)$ projected onto the $(\theta_1,\theta_2)$-coordinates.}
\label{fig:remain_region}
\end{figure}
\end{minipage}
~~~~~
\begin{minipage}{0.5\textwidth}
\begin{figure}[H]
\centering
\includegraphics[width =0.6\textwidth]{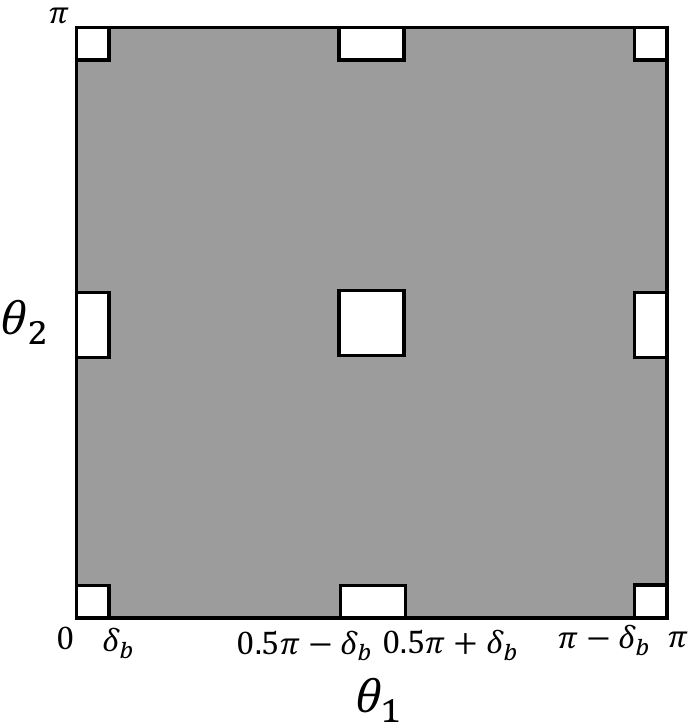}
\caption{The near
	band region $\bbN_b(\delta_b)$ projected onto the $(\theta_1,\theta_2)$-coordinates.}
\label{fig:band_region}
\end{figure}	
\end{minipage}}
\bigskip

To let the near band region cover $\bbN(\delta)\setminus \bbN_v(\delta_v)$, we define it as
\begin{align}
\bbN_b(\delta_b):=
&\left\{(\theta_1,\theta_2,\theta_3): \theta_i \in\left(\delta_b,\frac{\pi}{2}-\delta_b\right)\cup\left(\frac{\pi}{2}+\delta_b,\pi-\delta_b\right), i=1,2,3\right \}.
\label{eqn:near:band:region}
\end{align}
We plot the near band region $\bbN_b(\delta_b)$ projected onto the $(\theta_1,\theta_2)$-coordinates   in Fig.  \ref{fig:band_region}.

\begin{remark}
From  \eqref{def:remain:region} and \eqref{eqn:near:band:region}, we have   $\bbN_b(\delta_b)$  covers   $\bbN(\delta)\setminus \bbN_v(\delta_v)$ if $\delta_b\leq\min\{\delta_v,\delta\}$, or equivalently,
\begin{align}
\bbN(\delta)\subset  \bbN_b(\delta_b)\cup \bbN_v(\delta_v), \quad \text{if }
\delta_b\leq\min\{\delta_v,\delta\}.
\label{eqn:delta:b}
\end{align}	
\end{remark}

\paragraph{Controlling in Near Band Region}	
We start with the uniform upper-bound in Lemma \ref{lem:Parameterize}:
\begin{align}
&|F(\theta_1,\theta_2,\theta_3)|
\leq |\cos(\theta_1)\cos(\theta_2)\cos(\theta_3)|+ |\sin(\theta_1)\sin(\theta_2)\sin(\theta_3)|
+ \frac{4}{3}\tau(\log n)n^{-r_c}\nn\\
\leq& \frac{1}{3}(|\cos(\theta_1)|^3 +|\cos(\theta_2)|^3+|\cos(\theta_3)|^3 )
+\frac{1}{3}(|\sin(\theta_1)|^3+|\sin(\theta_2)|^3+|\sin(\theta_3)|^3)+\frac{4}{3}\tau(\log n) n^{-r_c}
\nn\\
\leq& \frac{1}{3}(|\cos(\theta_i)|^3 + |\sin(\theta_i)|^3)+\frac{2}{3}+\frac{4}{3}\tau(\log n) n^{-r_c},~\forall i\in\{1,2,3\}
\label{eqn:require:theta1:0}
\end{align}
where the first inequality follows from \eqref{eqn:bound:angular:form} in Lemma \ref{lem:Parameterize} (under Assumptions I-III and $r\leq n^{1.25-1.5 r_c}$ with $r_c\in(0,\frac{1}{6})$), the second inequality follows from  the inequality of arithmetic and geometric means, and the last one is a consequence of $|\sin(\theta)|^3 + |\cos(\theta)|^3 \leq 1$. So,   $|F(\theta_1,\theta_2,\theta_3)|<1$ in $\bbN_b(\delta_b)$ if
\begin{align}\label{eqn:require:theta1}
|\cos(\theta_i)|^3 + |\sin(\theta_i)|^3 < 1- 4\tau(\log n) n^{-r_c}
\end{align}
for \emph{some} $i \in\{1,2,3\}$. The final result is summarized in the following lemma, with the proof listed in Appendix \ref{App:proof:lem:middle:region}.
\begin{lemma}[Controlling in Near Band Region]\label{lem:near:band:region}
Under  Assumptions I, II, III,   if $r\leq n^{1.25-1.5 r_c}$   with $r_c\in(0,\frac{1}{6})$, then for sufficiently large $n$,
we have $|F(\theta_1,\theta_2,\theta_3)|<1$ in $\bbN_b(\delta_b)$ for $\delta_b=\sqrt{\frac{80\tau(\log n)}{3}}n^{-0.5 r_c}$.
\end{lemma}

\subsubsection{Combining the Near Vertex Region and Near Band Region}
Finally the \textsf{Angular-BIP} \eqref{eqn:angular:bip}  follows from Lemma \ref{lem:near:vertex:region}  and Lemma \ref{lem:near:band:region} if the union of the near vertex region $\bbN_v(\delta_v)$ and the near band region $\bbN_b(\delta_b)$   covers the near angular region $\bbN(\delta)$:
\[\bbN(\delta)\subset \bbN_v(\delta_v)\cup\bbN_b(\delta_b).\]
From \eqref{eqn:delta:b},
this happens when
\[
\delta_b\le\min\{\delta,\delta_v\},
\]
which is equivalent to
\begin{align}\label{eqn:cond:delta}
\delta_b\leq\delta,
\end{align}
since
$\delta_b=\sqrt{\frac{80\tau(\log n)}{3}}n^{-0.5 r_c}\ll\frac{\sqrt2-1}{3}=\delta_v.$

Then by Proposition \ref{pro:superset},   $q$ satisfies the \textsf{BIP} in $\calN_1(\delta)$. Similar results apply to all individual near region $\calN_p(\delta),$ for $p\in[r]$. Therefore we claim the \textsf{BIP} holds in the whole near region $\calN(\delta)=\bigcup_{p=1}^r\calN_p(\delta)$.
\begin{lemma}[Near-Region Bound]
\label{lem:near:region}
Under  Assumptions I, II, III,    if $r\leq n^{1.25-1.5 r_c}$   with $r_c\in(0,\frac{1}{6})$,
then for sufficiently large $n$,  the dual polynomial $q$ satisfies the {\textsf{BIP}} in $\calN(\delta)$ for any $\delta\ge\delta_b$.
\end{lemma}

\subsection{Combining the Far Region and Near Region}
Combining Lemma \ref{lem:far:region} (for far region)   and  Lemma \ref{lem:near:region} (for near region), we conclude that
the \textsf{BIP}  holds in the whole domain $\bbK$ if Assumptions I, II, III are satisfied and
\begin{align}
r&\leq \frac{n}{24\delta c^2} \text{ for } \delta\in[\delta_b, \frac{1}{24}] 
\quad\text{ and }\quad
r\leq n^{1.25-1.5 r_c} \text{ for }r_c\in(0,\frac{1}{6})\label{eqn:final:2}.
\end{align}
Then letting  $\delta=\delta_b$ (to maximize $r$) and $r_c=\frac{1}{8}$, the requirements \eqref{eqn:final:2} on $r$ are reduced to the desired bound \eqref{bound:r}:
$r\leq \frac{n^{17/16}}{32c^2\sqrt{15\tau(\log n)}}.$
The proof of Theorem \ref{thm:main} is completed.
\hfill$\square$

\section{Computational Method}\label{sec:computation}

Theorem \ref{thm:main} shows that when the  tensor factors $\{(\u_p^\star,\v_p^\star,\w_p^\star)\}_{p=1}^r$ satisfy  Assumptions I, II, III, we can recover the tensor decomposition of $r$ up to the order of $n^{17/16}$ by solving the convex, infinite-dimensional optimization  \eqref{eqn:method}. However, as a measure optimization problem, optimization problem~\eqref{eqn:method} is not directly solvable on a computer. In this section, we first propose a computational method  based on the popular Burer-Monteiro factorization method \cite{Burer:2003fg} and then test it by numerical experiments.

\begin{theorem}\label{pro:nonlinear}
Suppose the decomposition that achieves the tensor nuclear norm $\| \tT \|_{\ast}$
involves $r$ terms and $ \tilde{r} \geq r$, then $\| \tT \|_{\ast}$ is equal to the optimal value of the
following optimization:

\begin{align}
\minimize_{\{\u_p, \v_p, \w_p\}_{p=1}^{\tilde{r}}}  &
\sum_{p=1}^{\tilde{r}}\frac{1}{3} \left( \| \u_p \|_{2}^{3} + \| \v_p \|_{2}^{3} +
\| \w_p \|_{2}^{3}  \right)\ 
\st  \tT= \sum_{p=1}^{\tilde{r}} \u_p \otimes \v_p \otimes \w_p \label{eqn:nonlinear}
\end{align}
\end{theorem}

\begin{proof}
Suppose the tensor nuclear norm is achieved by the decomposition
\[
\tT = \sum_{p=1}^{r} \lambda_{p}^{\star} \u_{p}^{\star} \otimes \v_{p}^{\star}
\otimes \w_{p}^{\star}
.\]
Then we note that $\{{\lambda^\star_p}^{1/3}\u_p^\star, {\lambda^\star_p}^{1/3}\v_p^\star, {\lambda^\star_p}^{1/3}\w_p^\star\}_{p=1}^{\tilde{r}}$ forms a feasible solution to \eqref{eqn:nonlinear} when $\tilde{r} =r$. When $\tilde{r} >r$, we can zero-pad the remaining  factors $\{\u_p, \v_p, \w_p\}_{p=r+1}^{\tilde{r}}$. The objective
function value at this feasible solution is
$\frac{1}{3} (\sum_{p=1}^{\tilde{r}} 3 \lambda_{p}^{\star} ) = \| \tT \|_{\ast}$.
This shows that
$\|\tT\|_\ast$ is greater than the optimal value of \eqref{eqn:nonlinear}.

To show the other, suppose an optimal solution of \eqref{eqn:nonlinear} is $\{\u_p, \v_p, \w_p\}_{p=1}^{\tilde{r}}$. Define $\lambda_{p} := \| \u_p \|_{2} \| \v_p \|_{2} \| \w_p \|_{2}$, for $p\in[\tilde r].$
Then,
\[
\tT   = \sum_{p: \lambda_p \neq 0} \lambda_p \frac{\u_p}{\|\u_p\|_2} \otimes \frac{\v_p}{\|\v_p\|_2} \otimes \frac{\w_p}{\|\w_p\|_2}.
\]
By definition of the tensor nuclear norm \eqref{eqn:atomic:tensor}, we have
\begin{align*}
\| \tT \|_{\ast}
&\leq \sum_{p: \lambda_p \neq 0}  \lambda_{p}
= \sum_{p=1}^{\tilde r}  \lambda_{p}
=   \sum_{p=1}^{\tilde r} \| \u_p \|_{2}\| \v_p \|_{2}\| \w_p \|_{2}
\leq  \frac{1}{3} \sum_{p=1}^{\tilde{r}} \left[ \| \u_p \|_{2}^{3} + \| \v_p \|_{2}^3 + \| \w_p \|_{2}^3\right],
\end{align*}
which is the optimal value of \eqref{eqn:nonlinear}. Therefore, the optimal value of \eqref{eqn:nonlinear} is equal to $\| \tT \|_{\ast}$.
\end{proof}

Theorem \ref{pro:nonlinear} implies that when an upper bound on $r$ is known, we can solve the nonlinear (and non-convex) program \eqref{eqn:nonlinear} to compute the tensor nuclear norm (and obtain the corresponding decomposition). Despite the nonconvex nature of \eqref{eqn:nonlinear},  numerical simulations suggest that the ADMM approach \cite{boyd:2011bwa} has superior performance in solving \eqref{eqn:nonlinear}.
 
\section{Numerical Experiments and Beyond} 
Now we perform some numerical results to test the performance of the proposed  Burer-Monteiro factorization method. In particular, we will examine the phase transition of the rate of success for the ADMM implementation of the proposed Burer-Monteiro factorization approach \eqref{eqn:nonlinear}  with random initialization.  To illustrate the superiority of the proposed  tensor nuclear norm approach, we compare it  with the Least Squares formulation, that is, the L2 error minimization problem 
\begin{equation}
\minimize_{\{\u_p, \v_p, \w_p\}_{p=1}^{\tilde{r}}}\left\|\tT- \sum_{p=1}^{\tilde r} \u_{p} \otimes \v_{p} \otimes \w_{p}\right\|_{F}^{2}.
\label{eqn:LS}
\end{equation}
In the experiments, the $r$ tensor factors $\{(\u_p^\star, \v_p^\star, \w_p^\star)\}_{p=1}^r$ were generated following i.i.d. Gaussian distribution, and then each $\u_p^\star, \v_p^\star, \w_p^\star$ was normalized to have a unit norm. We set the coefficients  $\lambda_p^\star = (1 + \varepsilon_p^2)/2$, where $\varepsilon_p$ is chosen from the standard normal distribution, to ensure a minimal coefficient of at least $1/2$.  
With the generated ground-truth factors $\{(\u_p^\star,\v_p^\star,\w_p^\star)\}_{p=1}^r$ and coefficients $\{\lambda_p\}_{p=1}^r$, we generated the tensor $\tT=\sum_{p=1}^r  \lambda_p^\star \u_p^\star\otimes \v_p^\star \otimes \w_p^\star .$
To generate the phase transition plot, we varied the dimension $n$ and factor-number $r$, and for each fixed $(r, n)$ pair, 20 instances of such tensor were generated.   We then ran the ADMM algorithm to minimize \eqref{eqn:nonlinear}, and ran LBFGS to minimize the L2 error function \eqref{eqn:LS}, from the same random initialization. We remark that the global minimum value of the minimization \eqref{eqn:nonlinear} keeps the same for any $\tilde r\ge r$ but the global minimum solution doesn't. Therefore, to find all the true tensor factors, we choose $\tilde r=r$ in both methods.    For each instance, we declared success if the relative recovery error $\mathrm{Err}\left(\{(\widehat\u_p,\widehat\v_p,\widehat\w_p)\}_{p=1}^r\right)$ of the output tensor factors $\{(\widehat\u_p,\widehat\v_p,\widehat\w_p)\}_{p=1}^r$ (after removing sign and permutation ambiguities) is within $10^{-3}$ where
\[
\mathrm{Err}\left(\{(\widehat\u_p,\widehat\v_p,\widehat\w_p)\}_{p=1}^r\right):=\sum_{p=1}^r \left(\frac{\|\widehat{\u}_p-\u_p^\star\|_2}{\|\u^\star_p\|_2}+\frac{\|\widehat{\v}_p-\v_p^\star\|_2}{\|\v^\star_p\|_2}+\frac{\|\widehat{\w}_p-\w_p^\star\|_2}{\|\w^\star_p\|_2}\right).
\] 
We plot the experiment results in  Fig. \ref{fig:phase_transition}, which shows that the proposed method is clearly superior compared to the traditional Least Squares method. 

\begin{figure*}[ht!]
\centering
\includegraphics[width=0.95\textwidth]{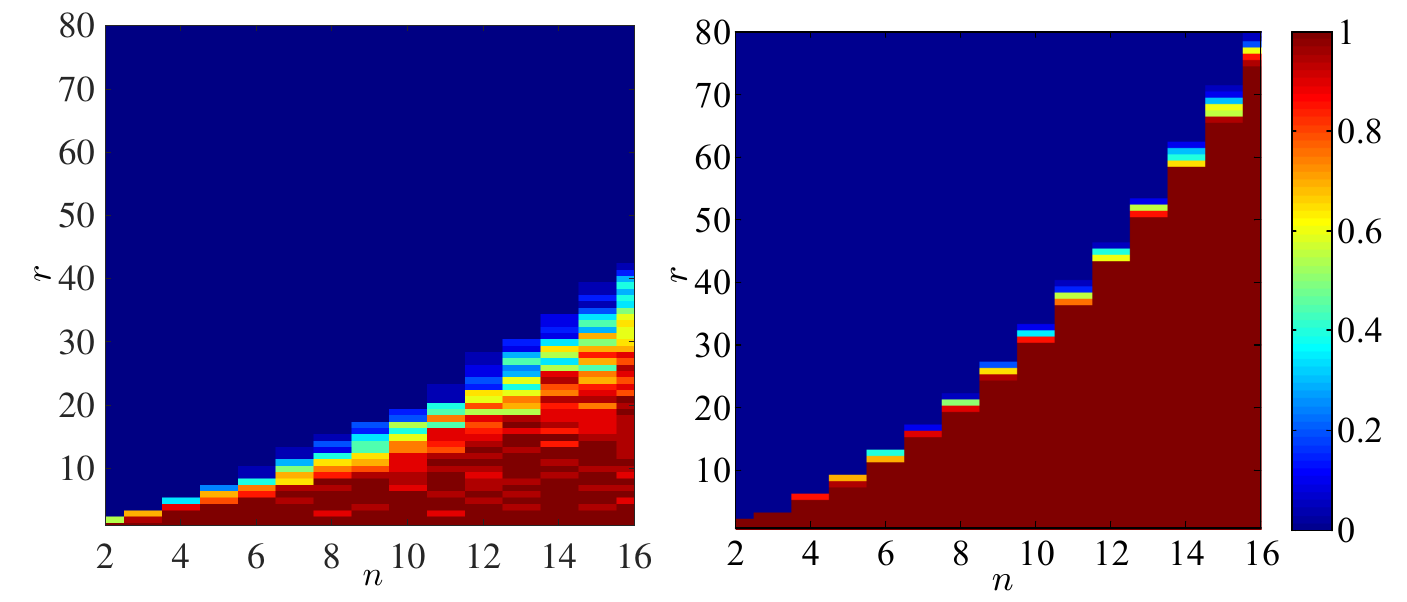}
	\caption{Rate of success using Least Squares method (Left) and ADMM implementation of \eqref{eqn:nonlinear} (Right) for tensor decomposition, respectively.}
	\label{fig:phase_transition}
\end{figure*}

Due to the nonconvexity nature of the two tensor decomposition formulations, 
we believe that the performance gain achieved by the proposed ADMM approach is because the optimization landscape of the Least Squares formulation of tensor decomposition is not as good as that of the tensor nuclear norm formulation \eqref{eqn:nonlinear}. We therefore conjecture that the tensor nuclear norm is crucial in flatting out the spurious local minima and high-order saddle points so that it helps to provide a benign optimization landscape, e.g., ``strict saddle property'', i.e., every critical point is either a strict saddle (where the Hessian has negative eigenvalues) or a global minimizer, see \cite{li2018non, zhu2018global, zhu2021global, li2017geometry, zhu2019distributed,li2020global, li2019alternating,li2019provable, li2017convex, li2022tensor,chi2019nonconvex} for more literature of landscape analysis. In contrast, the Least Squares formulation of tensor decomposition doesn't  satisfy the ``strict saddle property''. To verify this conjecture, we perform some preliminary analysis.  To simplify the notations and analysis, we consider the symmetric case as a first step. We believe the nonsymmetric case will have similar properties. More precisely, we consider the following L2 loss function
\begin{equation}
 g( \mU)=\frac{1}{6}\|\u_1\otimes^3+\u_2\otimes^3-\a_1\otimes^3-\a_2\otimes^3\|_F^2   
\end{equation}
where $\u\otimes^3:=\u\otimes\u\otimes\u$ and $ \mU:=[\u_1, \u_2]\in\mathbb{R}^{2\times 2}$. For this function, the critical points are given by the following equation:
\begin{equation}
\nabla g( \mU)= \mU[( \mU^\top  \mU)\odot ( \mU^\top  \mU)]- \mA[( \mA^\top \mU)\odot ( \mA^\top  \mU)]=\zero.   
\end{equation}
For simplicity, assume $ \mA=\eye$ the identity matrix. Then the stationary equation reduces to
\begin{equation}
 \mU[( \mU^\top  \mU)\odot ( \mU^\top  \mU)] =  \mU\odot  \mU   
 \label{eq:cri}
\end{equation}
Directly solving the above equation \eqref{eq:cri} (through \textsc{Mathematica}) generates three sets of local solutions.

\begin{description}
\item[\bf Case I]
\[\mU_1(x)=
\begin{bmatrix}
0 &~~ 0 \\
\sqrt[3]{1-x^3} &~~ x
\end{bmatrix} \text{ or } \begin{bmatrix}
\sqrt[3]{1-x^3} &~~ x\\
0 &~~ 0
\end{bmatrix} \text{ where $x\in\mathbb{R}$.}
\]
In this case, we then compute the eigenvalues of the Hessian matrix:
\[
{\vlambda}(\nabla^2g(U_1(x)))=
\begin{bmatrix}
0 \\
0 \\
\underbrace{\left(x^4-x^3 \sqrt[3]{1-x^3}+\sqrt[3]{1-x^3}\right)}_{\geq 3(2^{\frac{2}{3}})}
\\
3 \underbrace{\left(x^4-x^3 \sqrt[3]{1-x^3}+\sqrt[3]{1-x^3}\right)}_{\geq 3(2^{\frac{2}{3}})}
\end{bmatrix},
\]
where ${\vlambda}(\cdot)$ denotes the eigenvalue list of its argument.
We conclude that the first set of critical points $\mU_1(x)$ are neither strict saddle points (since Hessian matrix has no negative eigenvalues) nor global minima (since it is not a permuted version of Identify matrix). Therefore, the Least Squares formulation of tensor decomposition doesn't satisfy the ``strict saddle property''. In addition, since the Hessian doesn't have negative curvature at these critical points, the iterative algorithm such as gradient descent easily gets trapped by these points. This explains the relatively poor performance of the Least Squares formulation of tensor decomposition in Fig. \ref{fig:phase_transition}. 

\item[\bf Case II]
\[
\mU_2(x)=
\begin{bmatrix}
\sqrt[3]{0.25\, -x^3} &~~ x \\
\sqrt[3]{0.25\, -x^3} &~~ x
\end{bmatrix}
\]
For this set of critical points, there are no closed-form eigenvalues of its Hessian matrix. For convenience, we plot the four eigenvalues as a function of $x$.
\begin{center}
\includegraphics[width=.92\textwidth]{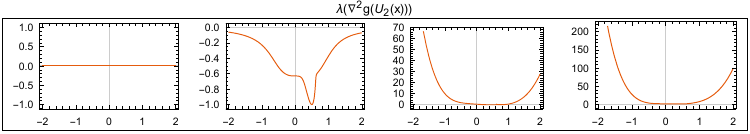}  
\end{center}
We therefore conclude that this set of critical points are strict saddle points at least when $x\in[-2,2]$ because the Hessian matrix has a negative eigenvalue for $x\in[-2,2]$.

\item[\bf Case III]
\[
\mU_3(x)=
\begin{bmatrix}
0 &~~ 0 \\
0 &~~ 0
\end{bmatrix}
\]
It is clear that the zero point is a special critical point and it is actually a high-order saddle point, because the Hessian matrix at zero point is also zero.
\end{description}
Although the $2$-by-2 case is very simple, it provides some evidence that the Least Squares formulation of tensor decomposition doesn't satisfy the ``strict saddle property'' and has high-order saddle points.

\section{Conclusion}\label{sec:conclusion}
By explicitly constructing a  dual certificate, we derive similar incoherence conditions (as the separation conditions in super-resolution problem) for a tensor
decomposition to achieve the tensor nuclear norm. This implies that the infinite dimensional total mass minimization can globally identify those decompositions satisfying the developed incoherence conditions. Computational method based on Burer-Monteiro factorization approach is used to solve the measure optimization.
Numerical experiments show
that the  Burer-Monteiro factorization approach achieves amazingly superior performance. Future
work will analyze the nonconvex landscape of  the   Burer-Monteiro factorization approach.

\appendix

\section{Proof of Lemma \ref{lem:relaxedcondition}} 
\label{App:proof:lem:relaxedcondition}

\begin{proof}
From the KKT conditions of the constrained optimization \eqref{eqn:mip:optmization}, we have the partial derivatives of its Lagrangian
\begin{align*}
\calL(\u,\v,\w,a,b,c) =&  q ( \u,\v,\w ) -a ( \| \u \|_{2}^{2} -1 ) -b ( \| \v \|_{2}^{2}-1 )  -c ( \| \w \|_{2}^{2} -1 )
\end{align*}
at $\u=\u_p^\star$, $\v=\v_p^\star$, and $\w=\w_p^\star$, $p=1,\ldots, r$, must vanish. Therefore,
\begin{equation}
\begin{aligned}
\frac{\partial \calL( \u_{p}^{\star} ,\v_{p}^{\star} ,\w_{p}^{\star},a,b,c )}{\partial\u}
& =  \frac{\partial q( \u_{p}^{\star} ,\v_{p}^{\star},\w_{p}^{\star})}{\partial \u} -2a\u_{p}^{\star} =0,\\
\frac{\partial \calL( \u_{p}^{\star} ,\v_{p}^{\star} ,\w_{p}^{\star},a,b,c )}{\partial\v}
& =  \frac{\partial q( \u_{p}^{\star} ,\v_{p}^{\star},\w_{p}^{\star})}{\partial \v} -2b\v_{p}^{\star} =0,\\
\frac{\partial \calL( \u_{p}^{\star} ,\v_{p}^{\star} ,\w_{p}^{\star},a,b,c )}{\partial\w}
& =  \frac{\partial q( \u_{p}^{\star} ,\v_{p}^{\star},\w_{p}^{\star} )}{\partial \w} -2c\w_{p}^{\star} =0.
\end{aligned}
\label{eqn:grad:zero}
\end{equation}
Hence, $2a=\langle \frac{\partial q( \u_{p}^{\star} ,\v_{p}^{\star},\w_{p}^{\star} )}{\partial \u}, \u_{p}^{\star} \rangle$, $2b=\langle \frac{\partial q( \u_{p}^{\star} ,\v_{p}^{\star},\w_{p}^{\star} )}{\partial \v}, \v_{p}^{\star} \rangle$, and $2c=\langle \frac{\partial q( \u_{p}^{\star} ,\v_{p}^{\star},\w_{p}^{\star} )}{\partial \w}, \w_{p}^{\star} \rangle$. Note that $q$ satisfies the Interpolation condition  and $\frac{\partial q ( \u,\v,\w )}{\partial \u ( i )} = \sum_{j,k} Q_{i j k} \v ( j ) \w ( k )$, we have that
\begin{align*}
2a&=   \sum_{i,j,k} Q_{i j k} \u_{p}^{\star} ( i )\v_{p}^{\star} ( j ) \w_{p}^{\star} ( k ) =  q ( \u_{p}^{\star} ,\v_{p}^{\star} ,\w_{p}^{\star} ) =1.
\end{align*}
That is $a=1/2$. With similar arguments, one can show that $b=c=1/2$. The conclusion of this lemma follows from \eqref{eqn:grad:zero}.
\end{proof}

\section{Proof of Lemma \ref{lem:normalequation}}
\label{App:proof:lem:normalequation}

\begin{proof}

First, the Lagrangian form of \eqref{eqn:minimum:norm} is
\begin{align*}
\calL(\tQ,\{\valpha_p^\star,\vbeta_p^\star,\vgamma_p^\star\}_{p=1}^r)
&=\frac{1}{2}\|\tQ\|_{F}^{2}-\sum_{p=1}^r\left(\tQ{\times}_1\valpha_p^\star{\times}_2\v_p^\star{\times}_3\w_p^\star+\tQ{\times}_1\u_p^\star{\times}_2\vbeta_p^\star{\times}_3\w_p^\star+\tQ{\times}_1\u_p^\star{\times}_2\v_p^\star{\times}_3\vgamma_p^\star\right)\\
&=\frac{1}{2}\|\tQ\|_{F}^{2}- \left\lg \tQ, \sum_{p=1}^r
\valpha_p^\star\otimes\v_p^\star\otimes\w_p^\star+ \u_p^\star\otimes\vbeta_p^\star\otimes\w_p^\star+
\u_p^\star\otimes\v_p^\star\otimes\vgamma_p^\star\right\rg
\end{align*}
with  the Lagrangian multipliers $\{\valpha_p^\star,\vbeta_p^\star,\vgamma_p^\star\}_{p=1}^r$ to be chosen such that $\tQ$  satisfies \eqref{eqn:minimum:energy}.
Then, by the KKT necessary conditions,  the solution of the least-norm problem (\ref{eqn:minimum:norm}) should satisfy
\begin{align*}
\zero=&\frac{\partial \calL(\tQ,\{\valpha_p^\star,\vbeta_p^\star,\vgamma_p^\star\}_{p=1}^r)}{\partial \tQ}
= \tQ-\sum_{p=1}^r \left( \valpha^\star_{p} \otimes \v_{p}^{\star} \otimes
\w_{p}^{\star} +\u_{p}^{\star} \otimes \vbeta^\star_{p} \otimes \w_{p}^{\star}
+\u_{p}^{\star} \otimes \v_{p}^{\star} \otimes \vgamma^\star_{p} \right).
\end{align*}	

\end{proof}

\section{{Proof of Lemma \ref{lem:Bound_coef}}}	\label{App:proof:lem:Bound_coef}

\begin{proof}

We need to find coefficients $\{\valpha^\star_p, \vbeta^\star_p, \vgamma^\star_p\}_{p=1}^r$ so that  
\[
\tQ = \sum_{p=1}^r \left( \valpha^\star_{p} \otimes \v_{p}^{\star} \otimes\w_{p}^{\star}
+\u_{p}^{\star} \otimes \vbeta^\star_{p} \otimes \w_{p}^{\star}
+\u_{p}^{\star} \otimes \v_{p}^{\star} \otimes \vgamma^\star_{p} \right)
\] 
satisfies
\begin{align}\label{eqn:tensor_eigenvalue}
&\tQ{\times}_2\v_p^\star
{\times}_3\w_p^\star = \u_p^\star,\quad \forall  p\in[r] ,\nonumber\\
& \tQ{\times}_1\u_p^\star
{\times}_3\w_p^\star = \v_p^\star,\quad \forall  p\in[r] ,\nonumber\\
& \tQ{\times}_1\u_p^\star
{\times}_2\v_p^\star = \w_p^\star,\quad \forall  p\in[r] .
\end{align}

\subsection{An Iteration Scheme}
We adopt the following {\em iterative scheme} to find such $\{\valpha_p^\star, \vbeta^\star_p, \vgamma^\star_p\}_{p=1}^r$:
\begin{align}
\valpha^{t+1}_q &= \valpha_q^t -\rho\left(\tQ_1^t{\times}_2\v_p^\star
{\times}_3\w_q^\star-\u_q^\star\right),\quad q\in[r],\nn\\
\vbeta^{t+1}_q  &= \vbeta_q^t -\rho\left(\tQ_2^t{\times}_1\u_p^\star
{\times}_3\w_q^\star-\v_q^\star\right), \quad q\in[r],\nn\\
\vgamma^{t+1}_q &= \vgamma_q^t -\rho\left(\tQ_3^t{\times}_1\u_p^\star
{\times}_2\v_q^\star-\w_q^\star\right),\quad q\in[r],\label{eqn:gamma:t+1}
\end{align}
initialized by $\valpha_q^0=\frac{1}{3}\u_q^\star$, $\vbeta_q^0=\frac{1}{3}\v_q^\star$, and $\vgamma_q^0=\frac{1}{3}\w_q^\star$ with $q\in [r]$.	Here the parameter $\rho$ is a step size  to be chosen later and the tensors
\begin{align}
\tQ_1^t:=& \sum_{p=1}^r \left( \valpha_{p}^t \otimes \v_{p}^{\star} \otimes
\w_{p}^{\star} +\u_{p}^{\star} \otimes \vbeta_{p}^\star \otimes \w_{p}^{\star}
+\u_{p}^{\star} \otimes \v_{p}^{\star} \otimes \vgamma_{p}^\star \right),
\nn\\
\tQ_2^t:=& \sum_{p=1}^r \left( \valpha_{p}^t \otimes \v_{p}^{\star} \otimes
\w_{p}^{\star} +\u_{p}^{\star} \otimes \vbeta_{p}^t \otimes \w_{p}^{\star}
+\u_{p}^{\star} \otimes \v_{p}^{\star} \otimes \vgamma_{p}^\star \right),
\nn\\
\tQ_3^t:=& \sum_{p=1}^r \left( \valpha_{p}^t \otimes \v_{p}^{\star} \otimes
\w_{p}^{\star} +\u_{p}^{\star} \otimes \vbeta_{p}^t \otimes \w_{p}^{\star}
+\u_{p}^{\star} \otimes \v_{p}^{\star} \otimes \vgamma_{p}^t\right).
\label{eqn:Qt}
\end{align}
Note that the above iterative scheme is for theoretical analysis only as we used $\{\valpha_p^\star, \vbeta^\star_p, \vgamma^\star_p\}_{p=1}^r$ in the definitions of $\tQ_1^t, \tQ_2^t$ and $\tQ_3^t$.

\subsection{Convergence Analysis}
We next establish the convergence of the iterations \eqref{eqn:gamma:t+1}. Plugging the tensor eigenvalue equations \eqref{eqn:tensor_eigenvalue} into \eqref{eqn:gamma:t+1} followed by subtracting the true solutions from both sides  yields  for $q\in[r]$
\begin{align}
\valpha^{t+1}_q - \valpha_q^\star
=& \valpha_q^t - \valpha_q^\star -\rho[\tQ_1^t-\tQ]{\times}_2\v_q^\star
{\times}_3\w_q^\star,\nn\\
\vbeta^{t+1}_q - \vbeta_q^\star
=& \vbeta_q^t - \vbeta_q^\star -\rho[\tQ_2^t-\tQ]{\times}_1\u_q^\star
{\times}_3\w_q^\star,\nn\\
\vgamma^{t+1}_q - \vgamma_q^\star
=& \vgamma_q^t - \vgamma_q^\star -\rho[\tQ_3^t-\tQ]{\times}_1\u_q^\star
{\times}_2\v_q^\star. \label{eqn:iterate}
\end{align}
Then plugging the definitions of $\tQ_1^t, \tQ_2^t,\tQ_3^t$ \eqref{eqn:Qt} into  \eqref{eqn:iterate} and
using the following matrix notations
\begin{align*}
\mA^t  &:=    \begin{bmatrix} \valpha^t_1, \cdots, \valpha^t_r\end{bmatrix},
\mA := \begin{bmatrix}\valpha_1^\star, \cdots, \valpha_r^\star\end{bmatrix},
\\
\mB^t & :=    \begin{bmatrix} \valpha^t_1, \cdots, \valpha^t_r\end{bmatrix},
\mB := \begin{bmatrix}\valpha_1^\star, \cdots, \valpha_r^\star\end{bmatrix}, 
\\
\mC^t &:=\begin{bmatrix} \vgamma^t_1, \cdots, \vgamma^t_r\end{bmatrix},~
\mC := \begin{bmatrix}\vgamma_1^\star, \cdots, \vgamma_r^\star\end{bmatrix},
\end{align*}
we have
\begin{align}
\mA^{t+1}-\mA
=& (\mA^t-\mA)(\mI-\rho\left[(\mV^\top \mV)\odot(\mW^\top \mW)\right]),
\nn\\
\mB^{t+1}-\mB
=& (\mB^t-\mB)\big(\mI-\rho[(\mU^\top \mU)\odot(\mW^\top \mW)]\big) -\rho \mV\left[((\mA^t-\mA)^\top \mU)\odot(\mW^\top \mW)\right],
\nn\\
\mC^{t+1}-\mC=& (\mC^t-\mC)(\mI-\rho[(\mU^\top \mU)\odot(\mV^\top \mV)])\nn 
\\
&-\rho  \mW\left\{[((\mA^t-\mA)^\top \mU)\odot(\mV^\top \mV)]
+[(\mU^\top \mU)\odot((\mB^t-\mB)^\top \mV)]\right\}.
\label{eqn:ABC:it}
\end{align}
Denoting  $e_a^t=\|\mA^t-\mA\|,e_b^t=\|\mB^t-\mB\|,e_c^t=\|\mC^t-\mC\|$ and
\[
\tilde{\rho}:=\rho\min
\left\{
\begin{matrix}
\lambda_{\min}((\mV^\top \mV)\odot(\mW^\top \mW))\\
\lambda_{\min}((\mU^\top \mU)\odot(\mW^\top \mW))\\
\lambda_{\min}((\mU^\top \mU)\odot(\mV^\top \mV))
\end{matrix}
\right\},
\]
it follows from \eqref{eqn:ABC:it} that
\begin{align}
e_a^{t+1}&\leq (1-\tilde\rho)e_a^t,\nn\\
e_b^{t+1}&\leq \rho\|\mU\|\|\mV\|\|\mW\|^2 e_a^t+(1-\tilde\rho)e_b^t,\nn\\
e_c^{t+1}&\leq \rho\|\mU\|^2\|\mV\|\|\mW\|e_a^t+\rho\|\mU\|^2\|\mV\|\|\mW\|e_b^t+(1-\tilde\rho)e_c^t,
\label{eqn:eaebec:it}
\end{align}
where we have used that $\|\mP\odot \mQ\|\leq \|\mP\otimes \mQ\| = \|\mP\|\|\mQ\|$.
Converting \eqref{eqn:eaebec:it} into matrix form gives
\begin{align*}
&\begin{bmatrix}
e_a^{t+1}\\
e_b^{t+1}\\
e_c^{t+1}
\end{bmatrix}
\leq
\begin{bmatrix}
1-\tilde{\rho} & 0 & 0\\
\rho\|\mU\|\|\mV\|\|\mW\|^2 & 1-\tilde{\rho} & 0\\
\rho \|\mU\|\|\mW\|\|\mV\|^2 & \rho\|\mU\|^2\|\mV\|\|\mW\| & 1-\tilde{\rho}
\end{bmatrix}
\begin{bmatrix}
e_a^t\\
e_b^t\\
e_c^t\\
\end{bmatrix},
\end{align*}
where the lower triangular  system matrix share the same value
\begin{align}
\eta&=1-\tilde{\rho}
\in\left[1-\rho\left(1+\frac{\kappa(\log n)\sqrt{r}}{n}\right),
1-\rho\left(1-\frac{\kappa(\log n)\sqrt{r}}{n}\right)\right]
\subset(0,1)
\label{eqn:eta}
\end{align}
where ``$\in$'' follows from applying \textsf{Weyl's inequality} to \eqref{cond:gram} in Assumption III and ``$\subset$'' holds for any $\rho\in \left(0,(1+\frac{\kappa(\log n)\sqrt{r}}{n})^{-1}\right)$.
 
Therefore, the error sequence $(e_a^t,e_b^t,e_c^t)$ is convergent to $(0,0,0)$ geometrically with a rate $\eta\in(0,1)$. Thus,
\[
\lim_{t\to\infty} (\mA^t,\mB^t,\mC^t)=(\mA,\mB,\mC).
\]

\subsection{Convergence of  $\{\|\mA^{t}-\mA^{t-1}\|\},\{\|\mB^{t}-\mB^{t-1}\|\},\{\|\mC^{t}-\mC^{t-1}\|\}$ } 
Subtracting the following two consecutive iterations for $\{\mA^t\}$ in \eqref{eqn:ABC:it}:
\begin{equation*}
\begin{aligned}
&\mA^{t+1}-\mA= (\mA^t-\mA)(\mI-\rho\left[(\mV^\top \mV)\odot(\mW^\top \mW)\right])\\
&\mA^{t}-\mA= (\mA^{t-1}-\mA)(\mI-\rho\left[(\mV^\top \mV)\odot(\mW^\top \mW)\right])
\end{aligned}
\quad\Longrightarrow
\mA^{t+1}-\mA^t =(\mA^t-A^{t-1})(\mI-\rho\left[(\mV^\top \mV)\odot(\mW^\top \mW)\right]).
\end{equation*}
Similar manipulations applied to $\{\mB^t\}$ and $\{\mC^t\}$ lead to
\begin{align*}
\mB^{t+1}-\mB^t=& (\mB^t-\mB^{t-1})(\mI-\rho\left[(\mU^\top \mU)\odot(\mW^\top \mW)\right])-\rho \mV\left[((\mA^t-\mA^{t-1})^\top \mU)\odot(\mW^\top \mW)\right],
\\
\mC^{t+1}-\mC^t=& (\mC^t-\mC^{t-1})(\mI-\rho\left[(\mU^\top \mU)\odot(\mV^\top \mV)\right])
\\
& -\rho \mW\left\{\left[((\mA^t-\mA^{t-1})^\top \mU)\odot(\mV^\top \mV)\right]+\left[(\mU^\top \mU)\odot((\mB^t-\mB^{t-1})^\top \mV)\right]\right\}
\end{align*}
Defining $\hat{e}_a^t=\|\mA^t-\mA^{t-1}\|,\hat{e}_b^t=\|\mB^t-\mB^{t-1}\|,\hat{e}_c^t=\|\mC^t-\mC^{t-1}\|$, we can get the same form as\eqref{eqn:eaebec:it} and therefore claim that $(\hat{e}_a^t,\hat{e}_b^t,\hat{e}_c^t)$ converge to $(0,0,0)$ geometrically with the same rate $\eta \in (0, 1)$ in \eqref{eqn:eta}.

\subsection{Controlling the Accumulative Errors} The geometric convergence of $\{\|\mC^t-\mC^{t-1}\|\}$ implies
\[\|\mC^t-\mC^{t-1}\| \leq \eta^{t-1}\|\mC^1-\mC^0\|\]
which implies that
\[
\|\mC^t-\mC^0\|
\leq \sum_{s=0}^{t-1} \|\mC^{s+1}-\mC^s\|
\leq  \sum_{s=0}^{t-1} \eta^s \|\mC^{1}-\mC^0\|
\leq  \frac{1}{1-\eta}\|\mC^1-\mC^0\|.
\]
 Let $t$ go to infinity:
\begin{align}
\|\mC - \mC^0\| &\leq \frac{1}{1-\eta} \|\mC^1-\mC^0\|.\label{C:C0}
\end{align}
We next bound $\|\mC^1-\mC^0\|$. From \eqref{eqn:gamma:t+1}, we have
\begin{align*}
&\vgamma^1_q-\vgamma^0_q=\rho (\tQ_3^0{\times}_1\u_q^\star{\times}_2\v_q^\star-\w_q^\star)
=\rho \left(\sum_{p=1}^r \lg \u_p^\star,\u_q^\star\rg\lg\v_p^\star,\v_q^\star\rg \w_p^\star- \w_q^\star\right)
\\
&\Longrightarrow
\mC^1 - \mC^0 = \rho \mW ((\mU^\top \mU)\odot (\mV^\top \mV)-\mI).
\end{align*}
Then from Assumptions II and III, we have
\begin{align}
\|\mC^1-\mC^0\|
\leq  \rho \|\mW\| \|(\mU^\top \mU)\odot (\mV^\top \mV)-\mI\|
\leq  \rho \bigg(1+c\sqrt{\frac{r}{n}}\bigg) \frac{\kappa(\log n)\sqrt{r}}{n}\label{C1:C0}.
\end{align}

\bigskip
\noindent{\bf Combining All }
Finally,  combining \eqref{eqn:eta}, \eqref{C:C0} and \eqref{C1:C0} and using $\mC_0=\frac{1}{3}\mW$,   we have
\begin{align*}
\left\|\mC - \frac{1}{3}\mW\right\|
\leq& \frac{1+c\sqrt{\frac{r}{n}}}{1- \frac{\kappa(\log n)\sqrt{r}}{n}}  \frac{\kappa(\log n)\sqrt{r}}{n}
\leq  2 \left(1+c\sqrt{\frac{r}{n}}\right) \frac{\kappa(\log n)\sqrt{r}}{n}
= 2\kappa(\log n)\bigg(\frac{\sqrt r}{n}+c{\frac{r}{n^{1.5}}}\bigg)
\end{align*}
where the second inequality follows from the assumption $r=o({n^2}/{\kappa(\log n)^2})$ which implies $1- \frac{\kappa(\log n)\sqrt{r}}{n}\geq \frac{1}{2}$ for a sufficiently large $n$.
Similar arguments and bounds apply to $\|\mA - \frac{1}{3}\mU\|$ and $\|\mB - \frac{1}{3}\mV\|$.
\end{proof}

\section{{Proof of Lemma \ref{lem:far:region}}}\label{App:proof:lem:Far_away}

\begin{proof}

The following lemma is required in the proof of Lemma \ref{lem:far:region}. Let us first admit Lemma~\ref{lem:p_norm} to prove Lemma \ref{lem:far:region}. Since $q$ is the sum of two parts given in 	\eqref{eqn:far:terms:0} and  \eqref{eqn:far:terms:2}, to bound $|q|$, we will control these parts separately.
 \begin{lemma}\label{lem:p_norm}
	Under  Assumptions I and II,  if $r\leq n^{1.25-1.5r_c} \text{ with } r_c\in(0,{1}/{6})$, then  for any integer $p\geq 3$,
	\begin{align*}
	\|\mU^\top \|_{2\to p}&\leq 1+\frac{1}{p}\tau(\log n) n^{-r_c}
	\end{align*}
	The same bounds hold for $\mV$ and $\mW$. Here, we define $\|\mH\|_{2\rightarrow p}:=\sup\{\|\mH \x\|_p: \x \in \S^{n-1}\}$.
\end{lemma}
\begin{proof}[Proof of Lemma \ref{lem:p_norm}]
See Appendix \ref{App:proof:lem:p_norm}.	
\end{proof}

\bigskip
\noindent{\bf  Bounding absolute value of \eqref{eqn:far:terms:0}:}
\begin{align*}
\sum_{p=1}^r |\langle \valpha^\star_{p}-\frac{1}{3}\u_p^\star ,\u \rangle \langle \v_{p}^{\star} ,\v
\rangle \langle \w_{p}^{\star} ,\w \rangle|
&\leq \sqrt{\sum_{p=1}^r \langle \valpha^\star_{p}-\frac{1}{3}\u_p^\star ,\u \rangle^2}
\sqrt{\sum_{p=1}^r \langle \v_{p}^{\star} ,\v
\rangle^2 \langle \w_{p}^{\star} ,\w \rangle^2}\\
&\leq \sqrt{\sum_{p=1}^r \langle \valpha^\star_{p}-\frac{1}{3}\u_p^\star ,\u \rangle^2}
\sqrt[4]{\sum_{p=1}^r \langle \v_{p}^{\star} ,\v
\rangle^4}
\sqrt[4]{\sum_{p=1}^r \langle \w_{p}^{\star} ,\w \rangle^4}\\
&=\|(\mA-\frac{1}{3}\mU)^\top\u\|_2\|\mV^\top\v\|_4\|\mW^\top\w\|_4\\
&\leq\|\mA-\frac{1}{3}\mU\| \|\mV^\top\|_{2\to 4}\|\mW^\top\|_{2\to 4}\\
&\leq
2\kappa(\log n)\left(\frac{\sqrt r}{n}+c\frac{r}{n^{1.5}}\right)(1+o(1))\\
&=o(1),
\end{align*}
where the last second line follows from Lemma \ref{lem:Bound_coef} and Lemma \ref{lem:p_norm} when $r\ll n^{1.25}$ (by letting $r_c$ in ``$r\ll n^{1.25-r_c}$" approach to zero). The last line holds for
$r\ll \frac{n^{1.5}}{\kappa(\log n)}.$

Similar bounds hold for the other two terms in \eqref{eqn:far:terms:0}.

\bigskip
\noindent{\bf Bounding the absolute value of \eqref{eqn:far:terms:2}:}
First of all, for any $(\u,\v,\w)\in\calF(\delta)$, there exists a division of $[r]=\Omega_u\cup\Omega_v\cup\Omega_w$ such that
\begin{equation}\label{eqn:divide:far}
\begin{aligned}
|\langle \u_p^\star,\u\rangle|&\leq \delta,\quad \forall p\in\Omega_u,\\
|\langle \v_p^\star,\v\rangle|&\leq \delta,\quad \forall p\in\Omega_v,\\
|\langle \w_p^\star,\u\rangle|&\leq \delta,\quad \forall p\in\Omega_w.
\end{aligned}
\end{equation}
We will denote by $\mU_{\Omega_u}$ the submatrix of $\mU$ forming from those columns of $\mU$ with indexes in $\Omega_u$. Similarly, we can define  $\mV_{\Omega_v}$ and $\mW_{\Omega_w}$. With these preparation, we have that
\begin{align*}
\sum_{p=1}^r | \langle \u_{p}^{\star} ,\u \rangle \langle\v_{p}^{\star} ,\v \rangle \langle \w_{p}^{\star} \w \rangle|
&= \sum_{p\in\Omega_u\cup\Omega_v\cup\Omega_w}| \langle \u_{p}^{\star} ,\u \rangle \langle\v_{p}^{\star} ,\v \rangle \langle \w_{p}^{\star} \w \rangle|\\
&\leq \delta(\|\mV_{\Omega_u}\|\|\mW_{\Omega_u}\|+\|\mU_{\Omega_v}\|\|\mW_{\Omega_v}\|+\|\mU_{\Omega_w}\|\|\mV_{\Omega_w}\|)\\
&\leq 3\delta\left(1+c \sqrt{\frac{r}{n}}\right)^2\\
&\leq 12\delta\max\{1,c^2r/n\}\\
&\leq \frac{1}{2},
\end{align*}
where the first inequality  follows from \eqref{eqn:divide:far} and
$ \sum_{p\in\Omega_u} |\langle\v_{p}^{\star} ,\v \rangle \langle  \w_{p}^{\star} ,\w \rangle| \leq \|\mV_{\Omega_u}\| \|\mW_{\Omega_u}\|$, etc.
The  second inequality uses the fact that the spectral norm of any submatrix is smaller than the original one and Assumption II.
The last inequality holds when $\delta\leq \frac{1}{24}$ and and $r\leq {n}/(24\delta c^2).$

\bigskip
\noindent{\bf Combining All }
Under  Assumptions I,  II, III, if $r\ll n^{1.25}$ and $r\leq \frac{n}{24\delta c^2}$ for $\delta\in(0,\frac{1}{24}]$,
we have $|q|\le o(1) +\frac{1}{2}<1$ in $\calF(\delta)$ for sufficiently large $n$.
\end{proof}

\subsection{Proof of Lemma \ref{lem:p_norm}}\label{App:proof:lem:p_norm}
The proof refines the one for Lemma 4 of \cite{anandkumar2015learning}. We only prove it for $\mU$ since the same arguments apply to $\mW$ and $\mV$.
We start with a general integer $p \geq 3$.
\begin{align}\label{eqn:E:1:a}
\|\mU^\top \|_{2\to p}= \sup_{\x\in\S^{n-1}} \|\mU^\top \x\|_p:=\|\mU^\top \x^\star\|_p
\end{align}
where we define $\x^\star\in \mathbb{S}^{n-1}$ to be the optimal solution of  $\sup_{\x\in\S^{n-1}} \|\mU^\top \x\|^p_p$. Further note that
\begin{align}\label{eqn:E:1:b}
\|\mU^\top \x^\star\|_p^p=\|\mU_S^\top \x^\star\|_p^p+\|\mU_{S^c}^\top \x^\star\|_p^p
\end{align}
where  $S$ denotes the indices of the largest (in absolute value) $L$ entries of $\mU^\top \x^\star$ and  $\mU_S$ denotes the column submatrix of $\mU$ indexed by $S$.
Similar notations apply to its complement set   $S^c=[r]\setminus S.$

\bigskip
\noindent{\bf  Bound the first term:}
\begin{align}
\|\mU_S^\top \x^\star\|_p^p
\leq \|\mU_S^\top \x^\star\|_2^2
\leq \|\mU_S\mU_S^\top \|
\leq 1+\sum_{i\in S\setminus \{j\}}|\langle \u_i, \u_j \rangle|
\leq 1+(L-1)\frac{\tau(\log n)}{\sqrt n}\label{eqn:first:bound}.
\end{align}
{\bf Note} this upper-bound is independent of $p$.
Here, the first inequality is because $|\u_i^{\star\top}\x^\star|\leq \|\u_i^\star\|_2\|\x^\star\|_2 = 1$ and the last second inequality follows from \textsf{Gershgorin's circle theorem}. Finally the last inequality  is from Assumption I and $L$ being the cardinality of the set $S$.

\bigskip
\noindent{\bf  Bound the second term:} First note that
\begin{align*}
\min_{i\in S} |\u_i^\top \x^\star|^2
\leq \frac{1}{L} \sum_{i\in S}|\u_i^\top \x^\star|^2
\leq \frac{1}{L}\|\mU_S\mU^\top _S\| \|\x^\star\|^2_2
\leq \frac{1}{L} (1+o(1))\leq \frac{2}{L}
\end{align*}
for sufficiently large $n$. The last second inequality follows from \eqref{eqn:first:bound} and an additional assumption on $L$
\begin{align}
(L-1)\frac{\tau(\log n)}{\sqrt n}=o(1).\label{eqn:cond:L}
\end{align}
We conclude that
\[
\max_{i\in S^c} |\u_i^\top \x^\star|^2 \leq \min_{i\in S} |\u_i^\top \x^\star|^2 \leq \frac{2}{L},
\]
since $S$ consists of the indices of the $L$ largest (in absolute value) elements of $\mU^\top\x^\star$.  As a consequence, we have
\begin{align}
\|\mU_{S^c}^\top \x^\star\|_p^p
&=    \sum_{i\notin S} |\u_i^\top  \x^\star|^p\nn
\\
&\leq \big(\max_{i\notin S} |\u_i^\top \x^\star|^{p-2}\big)\sum_{i\notin S} |\u_i^\top  \x|^2
&=    \big(\max_{i\notin S} |\u_i^\top \x^\star|^{p-2}\big)\|\mU_{S^c}^\top x^\star\|_2^2
\leq \bigg(\frac{2}{L}\bigg)^{\frac{p}{2}-1}\bigg(1+c \sqrt{\frac{r}{n}}\bigg)^2\label{eqn:general:bound}
\end{align}
where the last inequality follows from the fact that
$\|\mU_{S^c}^\top \x^\star\|_2^2\leq\|\mU_{S^c}\|^2\leq\|\mU\|^2\leq (1+c \sqrt{\frac{r}{n}})^2$ by Assumption II. Furthermore, since $(1+c \sqrt{\frac{r}{n}})^2 \le 4 \max\{1,c^2\frac{r}{n}\}$,   $c^2\frac{r}{n} \le c^2 n^{0.25-1.5r_c}$ from the condition of $r\leq n^{1.25-1.5r_c}$, and  $1\ll  c^2 n^{0.25-1.5r_c}$ for $r_c\in(0,1/6)$, we have
$(1+c \sqrt{\frac{r}{n}})^2 \le 4c^2 n^{0.25-1.5r_c}$ for $r_c\in(0,1/6)$. So from \eqref{eqn:general:bound}, we get
\begin{equation}\label{bound:p=3and4}
\|\mU_{S^c}^\top \x^\star\|_p^p \le 4\bigg(\frac{2}{L}\bigg)^{\frac{p}{2}-1}c^2 n^{0.25-1.5 r_c}.
\end{equation}

From \eqref{eqn:E:1:b}, \eqref{eqn:first:bound}, and \eqref{bound:p=3and4}, we have
\begin{equation*}
\|\mU^\top \x^\star\|_p^p \le 1+(L-1)\frac{\tau(\log n)}{\sqrt n} + 4\bigg(\frac{2}{L}\bigg)^{\frac{p}{2}-1}c^2 n^{0.25-1.5 r_c}.
\end{equation*}
By choosing
$
L=\left\lceil\frac{1}{2}n^{0.5-r_c}\right\rceil\Rightarrow
\begin{cases}L\leq \frac{1}{2}n^{0.5-r_c}+1\\
L\geq \frac{1}{2}n^{0.5-r_c}\end{cases}
$
(which satisfies the condition  \eqref{eqn:cond:L}), we have that
\begin{align*}
\|\mU^\top \x^\star\|_p^p
&\le 1+\frac{1}{2}\tau(\log n) n^{-r_c} + 4^{\frac{p}{2}} c^2 n^{(\frac{3}{4}-\frac{p}{4})+(\frac{p}{2}-\frac{5}{2})r_c}.
\end{align*}
Then from the assumptions $p\geq 3$ and $r_c\in(0,\frac{1}{6})$, we get
\begin{align}
\left(\frac{3}{4}-\frac{p}{4}\right)+\left(\frac{p}{2}-\frac{5}{2}\right)r_c\leq \left(\frac{3}{4}-\frac{p}{4}\right)6r_c+\left(\frac{p}{2}-\frac{5}{2}\right)r_c=(2-p)r_c\leq -r_c.
\label{eqn:r_c}
\end{align}
So, we have
$$
\|\mU^\top \x^\star\|_p^p \le 1+\left(\frac{1}{2}\tau(\log n)  + 4^{\frac{p}{2}} c^2 \right) n^{-r_c}.
$$
Since $4^{\frac{p}{2}} c^2\ll \frac{1}{2}\tau(\log n)$  and $(1+t)^{1/p} \le 1+\frac{1}{p}t$ for all $t\ge 0$, then
$$
\|\mU^\top \x^\star\|_p \le 1+ \frac{1}{p} \tau(\log n) n^{-r_c}
$$
holds for any $ p \ge 3$. This completes the proof since $\|\mU^\top \|_{2\to p}=\|\mU^\top \x^\star\|_p$ by \eqref{eqn:E:1:a}.

\section{Proof of Lemma \ref{lem:Parameterize}}\label{App:proof:lem:Parameterize}

\begin{proof}
We start by the angular dual polynomial 
\begin{align*}
q(\u(\theta_1),\v(\theta_2),\w(\theta_3))
=& \cos(\theta_1)\cos(\theta_2)\cos(\theta_3)
\\ 
&+q(\u_1^\star,\y,\z)\cos(\theta_1)\sin(\theta_2)\sin(\theta_3)+q(\x,\v_1^\star,\z)\sin(\theta_1)\cos(\theta_2)\sin(\theta_3)\nonumber\\
&+q(\x,\y,\w_1^\star)\sin(\theta_1)\sin(\theta_2)\cos(\theta_3)+q(\x,\y,\z)\sin(\theta_1)\sin(\theta_2)\sin(\theta_3).
\end{align*}
To bound $q$, we only need to bound the coefficients $q(\u_1^\star,\y,\z)$, $q(\x,\v_1^\star,\z)$, $q(\x,\y,\w_1^\star)$, and $q(\x,\y,\z)$.

We first show that $q(\u_1^\star,\y,\z)$, $q(\x,\v_1^\star,\z)$, and $q(\x,\y,\w_1^\star)$ are close to zero. To see this, we examine
\begin{align*}
&q(\x,\y,\w_1^\star)
= \sum_{p=1}^r [ \langle \valpha^\star_{p} ,\x \rangle \langle \v_{p}^{\star} ,\y
\rangle \langle \w_{p}^{\star} ,\w_1^\star \rangle
+ \langle \u_{p}^{\star} ,\x \rangle
\langle \vbeta^\star_{p} ,\y \rangle \langle \w_{p}^{\star} ,\w_1^\star \rangle
+ \langle
\u_{p}^{\star} ,\x \rangle \langle \v_{p}^{\star} ,\y \rangle \langle \vgamma^\star_{p}
,\w_1^\star \rangle]\\
=& \x^\top [\mA\diag(\mW^\top \w_1^\star)\mV^\top + \mU\diag(\mW^\top \w_1^\star)\mB^{\top}
+ \mU\diag(\mC^{\top} \w_1^\star)\mV^\top]\y\\
=& \x^\top \bigg(\mA\diag(\mW^\top \w_1^\star)\mV^\top-\frac{1}{3}\u_1^\star\v_1^\star
+ \mU\diag(\mW^\top \w_1^\star)\mB^{\top}-\frac{1}{3}\u_1^\star\v_1^\star
+ \mU\diag(\mC^{\top} \w_1^\star)\mV^\top -\frac{1}{3}\u_1^\star\v_1^\star\bigg)\y,
\end{align*}
since $\x\perp\u_1^\star,\y\perp\v_1^\star.$ This implies
\begin{align*}
|q(\x,\y,\w_1^\star)|
\leq &\left\|\mA\diag(\mW^\top \w_1^\star)\mV^\top-\frac{1}{3}\u_1^\star\v_1^\star\right\|
\\
&+ \left\|\mU\diag(\mW^\top \w_1^\star)\mB^{\top}-\frac{1}{3}\u_1^\star\v_1^\star\right\|
+ \left|\x^\top \left(\mU\diag(\mC^{\top} \w_1^\star)\mV^\top -\frac{1}{3}\u_1^\star\v_1^\star \right)\y\right|.
\end{align*}

We first bound $\left\|\mA\diag(\mW^\top \w_1^\star)\mV^\top-\frac{1}{3}\u_1^\star\v_1^\star\right\|$.
\begin{align}
\left\|\mA\diag(\mW^\top \w_1^\star)\mV^\top - \frac{1}{3}\u_1^\star \v_1^{\star\top}\right\|
\leq & \left\|\mA\diag(\mW^\top \w_1^\star)\mV^\top - \frac{1}{3} \mU\tmop{diag}(\mW^\top \w_1^\star)\mV^\top\right\| + \left\|\frac{1}{3}\mU\tmop{diag}(\mW^\top \w_1^\star)\mV^\top-\frac{1}{3}\u_1^\star \v_1^{\star\top}\right\|
\nn	\\
\leq & \left\|\mA - \frac{1}{3} \mU\right\| \|\tmop{diag}(\mW^\top \w_1^\star)\|\|\mV\| +\frac{1}{3} \|\mU\|\|\tmop{diag}(\mW^\top \w_1^\star-{\e_1})\|\mV^\top \|\nn\\
\leq & 2 \kappa(\log n) \left(\frac{\sqrt{r}}{n} + c \frac{r}{n^{1.5}}\right) \left(1 + c \sqrt{\frac{r}{n}}\right) +\frac{\tau(\log n)}{3\sqrt{n}} \bigg(1+c\sqrt{\frac{r}{n}}\bigg)^2	\nn\\
=&\left[2 \kappa(\log n) \frac{\sqrt{r}}{n} + \frac{\tau(\log n)}{3\sqrt{n}} \right]\left(1 + c \sqrt{\frac{r}{n}}\right)^2,
\nn
\end{align}
where the third inequality first
uses the facts $\|\tmop{diag}(\mW^\top \w_1^\star)\|=1$ and $\|\tmop{diag}(\mW^\top \w_1^\star-{\e_1})\|= \max_{p\neq 1}|\lg\w_p^\star,\w_1^\star\rg|$ and then follows from
Assumptions I and II and Lemma \ref{lem:Bound_coef}.

Similarly,
\begin{align*}
&\left\|\mU\diag(\mW^\top \w_1^\star)\mB^{\top}-\frac{1}{3}\u_1^\star\v_1^\star\right\|
\leq\left[2 \kappa(\log n) \frac{\sqrt{r}}{n} + \frac{\tau(\log n)}{3\sqrt{n}} \right]\left(1 + c \sqrt{\frac{r}{n}}\right)^2.
\end{align*}

The similar arguments also apply to bounding $|\x^\top(\mU\diag(\mC^{\top} \w_1^\star)\mV^\top -\frac{1}{3}\u_1^\star\v_1^\star)\y|$. Note that
\begin{align*}
\x^\top\left(\mU^\star\diag(\mC^{\top} \w_1^\star)\mV^{\top} - \frac{1}{3} \u_1^\star\v_1^{\star\top}\right)\y
=& \x^\top(\mU\diag((\mC-\mW/3)^\top \w_1^\star)\mV^{\top})\y
+\frac{1}{3}\x^\top(\mU\diag(\mW^{\top} \w_1^\star-\e_1)\mV^\top)\y
\end{align*}
and the first term can be rewritten as
\begin{align*}
\x^\top(\mU\diag((\mC-\mW/3)^\top \w_1^\star)\mV^\top)\y
&=\sum_{i=1}^r \x^\top\left( (\c_i-\w_i/3)^\top\w_1^\star\u_i\v_i^\top \right)\y
\\
&=\sum_{i=1}^r   (\x^\top\u_i)(\v_i^\top\y)(\c_i-\w_i/3)^\top\w_1^\star)
\\
&=  \x^\top\sum_{i=1}^r \left(\u_i(\v_i^\top\y)(\c_i-\w_i/3)^\top\right)\w_1^\star
\\
&= \x^\top\left(\mU\diag( \mV^\top\y)(\mC-\mW/3)^\top\right)\w_1^\star,
\end{align*}
and so
\begin{align*}
\left|\x^\top\left(\mU^\star\diag(\mC^{\top} \w_1^\star)\mV^{\top} - \frac{1}{3} \u_1^\star\v_1^{\star\top}\right)\y\right|
\leq& \|\mU\|\|\diag( \mV^{\top}\y)\|\|\mC-\mW/3\|+
\frac{1}{3}\|\mU\|\|\diag(\mW^{\top} \w_1^\star-\e_1)\|\mV^{\top}\|.
\end{align*}

Finally, we obtain
\begin{align*}
|q(\x,\y,\w_1^\star)|
\leq \left[6 \kappa(\log n) \frac{\sqrt{r}}{n} + \frac{\tau(\log n)}{\sqrt{n}} \right]\left(1 + c \sqrt{\frac{r}{n}}\right)^2
=&O\left(\frac{\kappa(\log n)\sqrt r}{n},\frac{\tau(\log n)}{\sqrt n},\frac{\kappa(\log n)r^{1.5}}{n^2},\frac{\tau(\log n)r}{n^{1.5}}\right)
\\
=&O\left( \frac{\kappa(\log n)}{n^{3/8+\frac{3}{4}r_c}},\frac{\tau(\log n)}{n^{5/8-\frac{3}{4}r_c}}, \frac{\kappa(\log n)}{n^{1/8+\frac{9}{4}r_c}},\frac{\tau(\log n)}{n^{\frac{1}{4}+1.5 r_c}}\right)
\\
= &O( \kappa(\log n)n^{-3r_c},\tau(\log n)n^{-3r_c})
=  o(n^{-2r_c})
\end{align*}
with the notation $O(f(n),g(n)):=\max\{O(f(n)),O(g(n))\}$.
The the last second line holds if $r\leq n^{1.25-1.5 r_c}$ and the last line follows from the assumption $r_c\in(0,1/6)$.

The same bound holds for $|q(\x,\v_1^\star,\z)|$ and $|q(\u_1^\star,\y,\z)|$.

The coefficient of the last term of \eqref{eqn:angular:form}  is $q(\x,\y,\z)$ and its absolute value is bounded by the tensor spectral norm of $\tQ$, and should be close to constant as $\tQ$ is close to $\sum_{p=1}^r \u_p^\star\otimes \v_p^\star \otimes \w_p^\star$, the spectral norm of which is $1+O(n^{-r_c})$ by the following lemma.

\begin{lemma}\label{lem:D:2}
	Under  Assumptions I and II, and if $r \leq n^{1.25-1.5 r_c}$ with $r_c\in(0,1/6)$,
	\begin{align*}
	\Bigg\|\sum_{p=1}^r \u_p^\star\otimes \v_p^\star \otimes \w_p^\star\Bigg\|\leq 1+\frac{5}{4}\tau(\log n) n^{-r_c}.
	\end{align*}
\end{lemma}

\begin{proof}[Proof of Lemma \ref{lem:D:2}]
	\begin{align*}
	\bigg\|\sum_{p=1}^r \u_p^\star\otimes \v_p^\star \otimes  \w_p^\star \bigg\|
	=&   \sup_{(\a,\b,\c)\in \mathbb{K}}
	\lg \mU^\top\a, (\mV^\top\b)\odot(\mW^\top\c)\rg
	\\
	\le&   \sup_{(\a,\b,\c)\in \mathbb{K}}
	\| \mU^\top\a\|_3 \|(\mV^\top\b)\odot(\mW^\top\c)\|_{3/2}
	\\
	\le&   \sup_{(\a,\b,\c)\in \mathbb{K}}
	\| \mU^\top\a\|_3\| \mV^\top\b\|_3\| \mW^\top\v\|_3
	\\
	\leq& \|\mU^\top \|_{2\to3}\|\mV^\top \|_{2\to3}\|\mW^\top \|_{2\to3}\nonumber\\
	\leq& \left(1+ \frac{1}{3}\tau(\log n) n^{-r_c}\right)^3
	\\
	=& 1+ \tau(\log n)n^{-r_c}+\frac{1}{3}\tau(\log n)^2n^{-r_c}+\frac{1}{9}\tau(\log n)^3n^{-3r_c}
	\leq 1+\frac{5}{4}\tau(\log n) n^{-r_c}\nn,
	\end{align*}
	where the first inequality follows from \textsf{H\"{o}lder's inequality} and the second follows from \textsf{Cauchy's inequality}. The fourth  follows from Lemma \ref{lem:p_norm} when
	$r \leq n^{1.25-1.5 r_c}$ with $r_c\in(0,\frac{1}{6}).$ The last  holds since
	$\frac{1}{3}\tau(\log n)^2n^{-r_c}+\frac{1}{9}\tau(\log n)^3n^{-3r_c}\ll\frac{1}{4}n^{-r_c}.$
\end{proof}

\bigskip

It remains to bound the difference between $\tQ$ and $\sum_{p=1}^r \u_p^\star\otimes
\v_p^\star \otimes \w_p^\star$:
\begin{align*}
&\bigg\|\mQ-\sum_{p=1}^r \u_p^\star\otimes\v_p^\star \otimes \w_p^\star\bigg\|
\\
&\leq \underbrace{\bigg\|\sum_{p=1}^r (\valpha_p^{\star} - \frac{1}{3} \u_p^\star)\otimes \v_p^\star \otimes \w_p^\star\bigg\|}_{\Pi_1}
+ \underbrace{\bigg\|\sum_{p=1}^r \u_p^\star \otimes (\vbeta_p^{\star} - \frac{1}{3} \v_p^\star) \otimes \w_p^\star\bigg\|}_{\Pi_2}
+\underbrace{\bigg\|\sum_{p=1}^r \u_p^\star\otimes  \v_p^\star \otimes (\vgamma_p^{\star} - \frac{1}{3} \w_p^\star)\bigg\|}_{\Pi_3}
\end{align*}
First we bound $\Pi_1$:
\begin{align}
\Pi_1
&= \sup_{(\a,\b,\c)\in\bbK}
\lg (\mA-\frac{1}{3}\mU)^\top\a,
(\mV^\top\b)\odot(\mW^\top\c)\rg
\nn\\
&\leq \sup_{(\a,\b,\c)\in \mathbb{K}} \|(\mA-\frac{1}{3}\mU)^\top\x\|_2
\|(\mV^\top\b)\odot(\mW^\top\c)\|_2
\nn\\
&\leq  \sup_{(\a,\b,\c)\in \mathbb{K}}\|(\mA-\frac{1}{3}\mU)^\top\x\|_2
\|(\mV^\top\b)\|_4\|(\mW^\top\c)\|_4
\nn\\
&\leq \|\mA-\frac{1}{3}\mU\|
\|\mV^\top\|_{2\to4}\|\mW^\top\|_{2\to4}
\nn\\
&\leq 2\kappa(\log n)\bigg(\frac{\sqrt r}{n}+c{\frac{r}{n^{1.5}}}\bigg)(1+o(1))
\leq  8\kappa(\log n)\max\Bigg\{\frac{\sqrt r}{n},c{\frac{r}{n^{1.5}}}\Bigg\}
\leq 8\kappa(\log n) n^{-3r_c}
= o(n^{-2r_c})
\nn
\end{align}
where the first and second inequalities follows from \textsf{Cauchy's inequality} 	
and the fourth inequality follows from
Lemma \ref{lem:Bound_coef} and Lemma  \ref{lem:p_norm} when $r\ll n^{1.25}$.
The last inequality follows by plugging $r\leq n^{1.25-1.5 r_c}$ with $r_c\in(0,\frac{1}{6})$.

The same bound also holds for $\Pi_2$ and $\Pi_3 $.

\bigskip
\noindent{\bf Combining All }  If $r\leq n^{1.25-1.5 r_c} \text{ with } r_c\in(0,{1}/{6})$, we have
\begin{equation}
\begin{aligned}
& |q(\u_1^\star,\y,\z)| = o(n^{-2r_c}),
\\
&|q(\x,\v_1^\star,\z)| = o(n^{-2r_c}),
\\
& |q(\x,\y,\w_1^\star)| = o(n^{-2r_c}),
\\
&|q(\x,\y,\z)| \leq 1+\frac{5}{4}\tau(\log n) n^{-r_c}+ o(n^{-2r_c})
\end{aligned}
\label{eqn:est:qxyz}
\end{equation}
which together with \eqref{eqn:angular:form}  gives
\begin{align*}
|q(\u(\theta_1),\v(\theta_2),\w(\theta_3))|
\leq& |\cos(\theta_1)\cos(\theta_2)\cos(\theta_3)|+ |\sin(\theta_1)\sin(\theta_2)\sin(\theta_3)| +\frac{5}{4}\tau(\log n) n^{-r_c}+ o(n^{-2r_c})\\
\leq& |\cos(\theta_1)\cos(\theta_2)\cos(\theta_3)|
+ |\sin(\theta_1)\sin(\theta_2)\sin(\theta_3)|+\frac{4}{3}\tau(\log n) n^{-r_c}
\end{align*}
where the last inequality follows from $o(n^{-2r_c})\ll \frac{1}{12}\tau(\log n)n^{-r_c}$.
\end{proof}

\section{Proof of Lemma \ref{lem:near:vertex:region}}\label{App:proof:lem:near:region}

\begin{proof}

Recall that	
\begin{equation}\label{eqn:expression:F}
\begin{aligned}
F(\theta_1, \theta_2, \theta_3)
=& \cos(\theta_1)\cos(\theta_2)\cos(\theta_3)
    +q(\u_1^\star,\y,\z)\cos(\theta_1)\sin(\theta_2)\sin(\theta_3)    
    \\
    &+q(\x,\v_1^\star,\z)\sin(\theta_1)\cos(\theta_2)\sin(\theta_3)
    +q(\x,\y,\w_1^\star)\sin(\theta_1)\sin(\theta_2)\cos(\theta_3)  
    \\
    &  +q(\x,\y,\z)\sin(\theta_1)\sin(\theta_2)\sin(\theta_3).
\end{aligned}
\end{equation}

The points of special interest are the eight vertices of the cube $[0,\pi]\times[0,\pi]\times[0,\pi]$, \emph{i.e.,}
\[\{(\theta_1,\theta_2,\theta_3): \theta_i\in\{0,\pi\}, i=1,2,3\}\]
which we classify into two sets:
\begin{enumerate}
	\item[(1)] The first set of vertices involve an even number of $\pi$: $(0,0,0),(0,\pi,\pi),(\pi,0,\pi),(\pi,\pi,0)$;
	 
	\item[(2)]  The second set of vertices involve an odd number of $\pi$: $(\pi,0,0),(0,\pi,0),(0,0,\pi),(\pi,\pi,\pi)$.
\end{enumerate}

\subsection{Control the First Vertex Set}
For the first set of points, we only show that
\begin{align*}
&F(\theta_1+\xi_1,\theta_2+\xi_2,\theta_3+\xi_3) \leq1,\quad
\forall \xi_i\in\bigg(-\frac{\sqrt 2-1}{3},\frac{\sqrt 2-1}{3}\bigg)
\bigcup
\bigg(\frac{\pi}{2}-\frac{\sqrt 2-1}{3},\frac{\pi}{2}+\frac{\sqrt 2-1}{3}\bigg)
\end{align*}
holds for $(\theta_1,\theta_2,\theta_3)=(0,0,0)$. The same arguments apply to the other cases $(\pi,0,\pi),(0,\pi,\pi),(\pi,\pi ,0)$ since \eqref{eqn:expression:F} implies
\begin{equation}\nn
\begin{aligned}
F(\xi_1,\xi_2,\xi_3)&=F(\xi_1,\pi+\xi_2,\pi+\xi_3)
=F(\pi+\xi_1,\xi_2,\pi+\xi_3)
=F(\pi+\xi_1,\pi+\xi_2,\xi_3)
\end{aligned}
\end{equation}
for all $\xi_1,\xi_2\,\xi_3\in\R.$

Let us apply the first-order Taylor expansion to $F(\theta_1, \theta_2, \theta_3)$ over some smaller cube
$[-\theta_0, \theta_0]\times[-\theta_0, \theta_0]\times[-\theta_0, \theta_0]$ with $\theta_0 \in (0, \pi/2)$
to be determined later,
\begin{align*}
F(\theta_1, \theta_2, \theta_3)
=&     F(0,0,0) + \vtheta^\top \nabla F(\xi_1,\xi_2,\xi_3)
\geq  1 - \|\vtheta\|_1\sup_{|\xi_1|,|\xi_2|,|\xi_3|\leq \theta_0}\|\nabla F(\xi_1,\xi_2,\xi_3)\|_\infty
,\end{align*}
where $\vtheta = \begin{bmatrix}
\theta_1&
\theta_2&
\theta_3
\end{bmatrix}^\top$.
Since
\begin{align*}
\frac{\partial}{\partial \theta_1} F(\xi_1,\xi_2,\xi_3)
= &-\sin(\xi_1)\cos(\xi_2)\cos(\xi_3)
  -q(\u_1^\star,\y,\z)\sin(\xi_1)\sin(\xi_2)\sin(\xi_3)  \\
  &+q(\x,\v_1^\star,\z)\cos(\xi_1)\cos(\xi_2)\sin(\xi_3)  +q(\x,\y,\w_1^\star)\cos(\xi_1)\sin(\xi_2)\cos(\xi_3) 
  \\
  & +q(\x,\y,\z)\cos(\xi_1)\sin(\xi_2)\sin(\xi_3),
\end{align*}	
we have
\begin{align}
\bigg|\frac{\partial}{\partial \theta_1} F(\xi_1,\xi_2,\xi_3)\bigg|
\leq& |\sin(\theta_0)|+ o(1) (|\sin(\theta_0)|^3+2|\sin(\theta_0)|)
+(1+o(1))|\sin(\theta_0)|^2
\nn\\
\leq& |\sin(\theta_0)|+|\sin(\theta_0)|^2+o(1)
\leq 3|\sin(\theta_0)|\nn
\end{align}
where the first inequality follows from  \eqref{eqn:est:qxyz}, and so
\begin{equation}
\begin{aligned}
|q(\u_1^\star,\y,\z)|&=o(1),
\\
|q(\x,\v_1^\star,\z)|&=o(1), 
\\
|q(\x,\y,\w_1^\star)|&=o(1), 
\\
|q(\x,\y,\z)|&=1+o(1)
\end{aligned}
\label{eqn:q:xyz:o1}
\end{equation}
under Assumptions I-III and $r\ll n^{1.25}$ (by letting $r_c$ in ``$r\ll n^{1.25-r_c}$" approach to zero).
The last inequality uses the facts that $|\sin(\theta_0)|^2\leq |\sin(\theta_0)|$
and $o(1)\leq |\sin(\theta_0)|$ for sufficiently large $n$.
The same bound holds for $\big|\frac{\partial}{\partial \theta_2} F(\xi_1,\xi_2,\xi_3)\big|$ and
$\big|\frac{\partial}{\partial \theta_3} F(\xi_1,\xi_2,\xi_3)\big|$.
We therefore have
\begin{align}
F(\theta_1, \theta_2, \theta_3) & \geq 1 - 3\|\vtheta\|_1|\sin(\theta_0)|
\geq 1-9\theta_0^2. \label{eqn:lower:F}
\end{align}

Let us recall the integral form of  the second-order Taylor expansion of $F(\theta_1,\theta_2,\theta_3)$:
\begin{align*}
F(\theta_1, \theta_2, \theta_3)
= F(0,0,0) &+ \vtheta^\top \nabla F(0,0,0)+ \int_{0}^1 \frac{t^2}{2}\vtheta^\top \nabla^2F(t\theta_1,t\theta_2,t\theta_3)\vtheta  \dif t
\end{align*}
As a consequence of the construction process of the dual polynomial, we have $F(0,0,0)=1$ and $\nabla F(0,0,0)=0$, implying
$$
F(\theta_1, \theta_2, \theta_3)  =1 + \int_{0}^1 \frac{t^2}{2}\vtheta^\top \nabla^2F(t\theta_1,t\theta_2,t\theta_3)\vtheta \dif t
$$
Therefore, as long as the Hessian matrix $\nabla^2F$  is negative definite over the region $[-\theta_0,\theta_0]^3$ for some  $\theta_0>0$,
then $F(\theta_1, \theta_2, \theta_3)\leq 1$ for any $(\theta_1, \theta_2, \theta_3)\in[-\theta_0,\theta_0]^3$ with equality holds only if $(\theta_1, \theta_2, \theta_3)=(0,0,0)$.

We next estimate the Hessian matrix $\nabla^2 F(\xi_1,\xi_2,\xi_3)$. Direct computation gives
\begin{align*}
\nabla^2 F(\xi_1,\xi_2,\xi_3)
=   & \begin{bmatrix}
-F(\xi_1,\xi_2,\xi_3) & * & *\\
* & -F(\xi_1,\xi_2,\xi_3) & * \\
* & * & -F(\xi_1,\xi_2,\xi_3)
\end{bmatrix}
\end{align*}
whose off-diagonal elements are nonsymmetric partial derivatives of $F$, for example,
\begin{align*}
\frac{\partial^2}{\partial \theta_1\partial \theta_2} F(\xi_1,\xi_2,\xi_3)
= \sin(\xi_1) \sin(\xi_2)\cos(\xi_3)
 &- q(\u_1^\star,\y,\z)\sin(\xi_1)\cos(\xi_2) \sin(\xi_3)
 \\
 &+ q(\x,\y,\w_1^\star) \cos(\xi_1)\cos(\xi_2)\cos(\xi_3)
 \\
& - q(\x,\v_1^\star,\z) \cos(\xi_1)\sin(\xi_2)\sin(\xi_3) \\
&+ q(\x,\y,\z) \cos(\xi_1)\cos(\xi_2)\sin(\xi_3),
\end{align*}
which implies by \eqref{eqn:q:xyz:o1} that
for any $|\xi_i|\le \theta_0,i=1,2,3$,
\begin{align*}
\bigg|\frac{\partial^2}{\partial \xi_1\partial \xi_2} F(\xi_1,\xi_2,\xi_3)\bigg|
\leq& |\sin(\theta_0)|^2 + o(1)(1+2|\sin(\theta_0)|^2) + (1+o(1))|\sin(\theta_0)| 
\\
\leq& |\sin(\theta_0)|+|\sin(\theta_0)|^2+o(1)
\leq 3|\sin(\theta_0)|.  \nn
\end{align*}
The same bound holds for other mixed partial derivatives
$ \big|\frac{\partial^2}{\partial \xi_i\partial \xi_j} F(\xi_1,\xi_2,\xi_3)\big|$ with $i,j=1,2,3$ and $i\neq j$.

To make $\nabla^2 F(\xi_1,\xi_2,\xi_3)$ negative definite, by \textsf{Gershgorin's circle theorem} and the bound \eqref{eqn:lower:F}, we only need
\begin{align*}
-F(\xi_1,\xi_2,\xi_3)+6|\sin(\theta_0)| \leq -1+9\theta_0^2+ 6\theta_0 < 0
\end{align*}
which holds for any $\theta_0\in(\frac{-\sqrt 2-1}{3},\frac{\sqrt 2-1}{3})$, including
$(\frac{-\sqrt 2+1}{3} ,\frac{\sqrt 2-1}{3})$. This completes the first part of the proof.

\subsection{Control the Second Vertex Set}
Similarly as before, we first show
\[
F(\pi+\xi_1,\pi+\xi_2,\pi+\xi_3)<0,~\forall |\xi_i|<\frac{\sqrt 2-1}{3}.
\]
It follows from the intermediate result  \eqref{eqn:lower:F}:
\[
F(\xi_1,\xi_2,\xi_3)  \geq 1-9\theta_0^2>0,~\forall  |\xi_i|\leq\theta_0
\]
by recognizing that $F(\pi+\xi_1,\pi+\xi_2,\pi+\xi_3)=-F(\xi_1,\xi_2,\xi_3), \forall \xi_1,\xi_2,\xi_3$ and choosing  $\theta_0={(\sqrt 2-1)}/{3}.$

Finally, we claim the same conclusion applies to the remaining three cases since
\[
F(\pi+\xi_1,\pi+\xi_2,\pi+\xi_3)=F(\pi+\xi_1, \xi_2, \xi_3)
=F( \xi_1,\pi+\xi_2, \xi_3)
=F( \xi_1, \xi_2,\pi+\xi_3)
\]
for all $\xi_1,\xi_2,\xi_3\in\R.$
\end{proof}

\section{Proof of Lemma \ref{lem:near:band:region}}\label{App:proof:lem:middle:region}

\begin{proof}

First, solve for $\theta$  such that
\begin{align}\label{eqn:sin:cos:bound}
|\cos(\theta)^3| + |\sin(\theta)|^3 < 1- 4\tau(\log n) n^{-r_c}.
\end{align}
To this end, we define $f(\theta):=|\cos(\theta)^3| + |\sin(\theta)|^3$ for $\theta \in [0, \pi]$. It can be verified directly that $f$ is symmetric around $\frac{\pi}{2}$ on $[0, \pi]$, symmetric around $\frac{\pi}{4}$ on $[0, \frac{\pi}{2}]$, and strictly decreasing on $[0, \frac{\pi}{4}]$. Since $1- 4\tau(\log n) n^{-r_c}\in(0,1)$, there exists a unique $\varpi \in (0, \frac{\pi}{4})$ such that $f(\varpi)=1- 4\tau(\log n) n^{-r_c}\in(0,1)$. Thus the inequality \eqref{eqn:sin:cos:bound} holds on $(\varpi, \frac{\pi}{2}-\varpi) \cup (\frac{\pi}{2}+\varpi, \pi-\varpi)$.

To have an approximation of $\varpi$, we need the following lemma.

\begin{lemma}\label{lem:fact}
	Let  $f$ and $g$ be any two real functions with  $g$ being strictly decreasing in some interval $(\alpha,\beta)$ and satisfying $g(x)\geq f(x),\forall x\in(\alpha,\beta)$. Suppose both equations $f(x)=b$ and $g(x)=b$ admit one root in $[\alpha,\beta]$, denoted by $x_f$ and $x_g$ respectively. Then  $x_g\geq x_f$.
\end{lemma}
\begin{proof}[Proof of Lemma \ref{lem:fact}]
	Since $g(x)>g(x_f)\geq f(x_f)=b$ for any $x\in[\alpha,x_f)$,  $g(x_g)=b$ could only happen within $[x_f,\beta]$.
\end{proof}

We recognize that
\begin{align}\label{eqn:upper:bound}
f(\theta)\leq 1-\frac{3}{20} \theta^2 ,\text{ for } \theta \in [0,\pi/4]
\end{align}
and $g(\theta):=1-\frac{3}{20}\theta^2$ is strictly deceasing $[0,\pi/4]$. Clearly,
$$
\delta_b:=\sqrt{\frac{80\tau(\log n)}{3}}n^{-0.5 r_c}
$$
is the root of $g(\theta)=1- 4\tau(\log n) n^{-r_c}$ over the interval $[0, \frac{\pi}{4}]$. By Lemma~\ref{lem:fact}, $\delta_b \ge \varpi$. Therefore, \eqref{eqn:sin:cos:bound} holds on $(\delta_b, \frac{\pi}{2}-\delta_b) \cup (\frac{\pi}{2}+\delta_b, \pi-\delta_b)$. By \eqref{eqn:require:theta1:0}, we obtain
$F(\theta_1,\theta_2,\theta_3)<1 \text{ for }(\theta_1,\theta_2,\theta_3)\in\bbN_b(\delta_b).
$
\end{proof}

	\subsection{Proof of (\ref{eqn:upper:bound})}
Showing \eqref{eqn:upper:bound}
is equivalent to showing
\begin{align}
\sin^3(x)+\cos^3(x)\leq 1-\frac{3}{20}x^2,~\forall x\in[0,\pi/4]
\label{eqn:sin:cos:final}
\end{align}
since $\sin(x),\cos(x)>0$ for $x\in[0,\pi/4]$. Before moving on, we need the following lemma to prove \eqref{eqn:sin:cos:final}.
	\begin{lemma}\label{pair}
		The following inequality
		\begin{equation}\label{eq1}
		\frac{(3^{2n-1}-3)}{4\cdot (2n-1)!}x^{2n-1}+ \frac{(3^{2n}+3)}{4\cdot (2n)!}x^{2n}-\frac{(3^{2n+1}-3)}{4\cdot (2n+1)!}x^{2n+1}- \frac{(3^{2n+2}+3)}{4\cdot (2n+2)!}x^{2n+2} \ge 0
		\end{equation}
		holds for all $x \in [0, \pi/4]$ and $n \ge 2$,
	\end{lemma}
\begin{proof}
	Let $p$ equal the expression on the left side of Equation~\eqref{eq1}. A simplification on $p$ yields
	\[
	p(x)= q_1(x) \frac{x^{2n-1}}{4(2n-1)!} + q_2(x)\frac{x^{2n+2}}{4(2n)!},
	\]
	where $\displaystyle q_1(x) = (3^{2n-1}-3) -  \frac{3^{2n+1}-3}{2n(2n+1)}x^2$ and $\displaystyle q_2(x) = (3^{2n}+3) - \frac{3^{2n+2}+3}{(2n+1)(2n+2)}x^2$.
	 
	As functions of $x$, $q_1$ and $q_2$ have roots at
	\[
	\pm \sqrt{\frac{2n(2n+1)(3^{2n-1}-3)}{3^{2n+1}-3}} \quad \text{and} \quad \pm\sqrt{\frac{(2n+1)(2n+2)(3^{2n}+3)}{3^{2n+2}+3}},
	\] respectively, provided $n\geq 1$.  
	
	Since
	\begin{align*}
	   10(3^{2n-1}-3) &\ge 3^{2n+1}-3 , \text{ for all }n\geq 2,
	   \\
	   9(3^{2n}+3)&>(3^{2n+2}+3), \text{ for all }n\geq 2,
	\end{align*}
it follows that the positive root of $q_1$ satisfies
	\[ \sqrt{\frac{2n(2n+1)(3^{2n-1}-3)}{3^{2n+1}-3}} \ge \sqrt{\frac{2n(2n+1)}{10}} > \sqrt{2} > \frac{\pi}{4}, \text{ for }n\geq 2,\]
	and the positive root of $q_2$ satisfies
	\[ \sqrt{\frac{(2n+1)(2n+2)(3^{2n}+3)}{3^{2n+2}+3}} > \sqrt{\frac{(2n+1)(2n+2)}{9}} > \sqrt{\frac{10}{3}} > \frac{\pi}{4}, \text{ for }n\geq 2. \]
	Therefore both $q_1$ and  $q_2$ are positive  on $[0,\pi/4]$ for all $n\geq 2$, and Equation~\eqref{eq1} holds.
\end{proof}

\begin{lemma}The following statement
	$$
	\sin^3(x) + \cos^3 (x) \le 1-\frac{3}{20}x^2
	$$
	holds for all $x\in [0, \frac{\pi}{4}]$.
\end{lemma}

\begin{proof}
	Recall that
	\begin{align*}
	 \sin^3(x) &= \frac{1}{4}\left(3\sin(x) - \sin(3x)\right),
	 \\
	 \cos^3(x) &= \frac{1}{4}\left(3\cos(x) + \cos(3x)\right).
	\end{align*}
 Therefore,
	\begin{align*}
	&\sin^3(x) = x^3 + \sum_{n=5}^\infty(-1)^{n} \frac{3^{2n-1}-3}{4(2n-1)!}x^{2n-1},	
	\\
	&\cos^3(x) = 1 - \frac{3}{2}x^2 + \frac{7}{8}x^4 +
	\sum_{n=3}^\infty (-1)^n\frac{3^{2n}+3}{4(2n)!}x^{2n}.
	\end{align*}
	Thus
	\[
	\sin^3(x) + \cos^3(x) \leq 1- \frac{3}{2}x^2 +  x^3  + \frac{7}{8}x^4,
	\]
	for all $x \in [0, \pi/4]$ since by Lemma~\ref{pair}
	\begin{align*}
	&	\sum_{n=3}^\infty(-1)^{n} \frac{3^{2n-1}-3}{4(2n-1)!}x^{2n-1} +  \sum_{n=3}^\infty (-1)^n\frac{3^{2n}+3}{4(2n)!}x^{2n} \\
	=& -\sum_{n=3, \; n\text{ odd}}^\infty \left(\frac{3^{2n-1}-3}{4(2n-1)!}x^{2n-1} +  \frac{3^{2n}+3}{4(2n)!}x^{2n} - \frac{3^{2n+1}-3}{4(2n+1)!}x^{2n+1} - \frac{3^{2n+2}}{4(2n+2)!}x^{2n+2}\right) 
	\\
	\le& 0.
	\end{align*}
	
	Finally, note that
	\[1-\frac{3}{2}x^2+x^3+\frac{7}{8}x^4 = 1-\frac{3}{20}x^2 + x^2h(x)\]
	with
	\[
	h(x)=-\frac{27}{20}+x+\frac{7}{8}x^2\ge 0 \text{ when }x\in[0,\pi/4].
	\]
\end{proof}

\section*{Funding}
This work was supported by the National Science Foundation [DMS-1913039 to L.S., CCF-2203060, CCF-2106834 to G.T.].

\section*{Data Availability Statements}
No new data were generated or analysed in support of this research.

\bibliographystyle{plain}
\bibliography{nonconvex_overcomplete}

\begin{thebibliography}{10}

\bibitem{anandkumar2015learning}
Animashree Anandkumar, Rong Ge, and Majid Janzamin.
\newblock Learning overcomplete latent variable models through tensor methods.
\newblock In {\em Conference on Learning Theory}, pages 36--112. PMLR, 2015.

\bibitem{anandkumar2017analyzing}
Animashree Anandkumar, Rong Ge, and Majid Janzamin.
\newblock {Analyzing tensor power method dynamics in overcomplete regime}.
\newblock {\em Journal of Machine Learning Research}, 18(22):1--40, 2017.

\bibitem{Barak:2014vw}
Boaz Barak, Jonathan~A Kelner, and David Steurer.
\newblock {Dictionary learning and tensor decomposition via the sum-of-squares
  method}.
\newblock In {\em Proceedings of the forty-seventh annual ACM symposium on
  Theory of computing}, pages 143--151. ACM, 2015.

\bibitem{Barvinok:2002vr}
Alexander Barvinok.
\newblock {\em {A Course in Convexity}}.
\newblock American Mathematical Soc., 2002.

\bibitem{Bendory:2014tl}
Tamir Bendory, Shai Dekel, and Arie Feuer.
\newblock {Super-resolution on the sphere using convex optimization}.
\newblock {\em Signal Processing, IEEE Transactions on}, 63(9):2253--2262,
  2015.

\bibitem{Liu:2013bh}
Johann~A Bengua, Ho~N Phien, Hoang~Duong Tuan, and Minh~N Do.
\newblock Efficient tensor completion for color image and video recovery:
  Low-rank tensor train.
\newblock {\em IEEE Transactions on Image Processing}, 26(5):2466--2479, 2017.

\bibitem{Bhaskar:2013ki}
B~N Bhaskar, Gongguo Tang, and B~Recht.
\newblock {Atomic norm denoising with applications to line spectral
  estimation}.
\newblock {\em Signal Processing, IEEE Transactions on}, 61(23):5987--5999,
  2013.

\bibitem{boyd:2011bwa}
Stephen Boyd.
\newblock {Distributed optimization and statistical learning via the
  alternating direction method of multipliers}.
\newblock {\em Foundations and Trends in Machine Learning}, 3(1):1--122, 2011.

\bibitem{Burer:2003fg}
Samuel Burer and Renato D~C Monteiro.
\newblock {A nonlinear programming algorithm for solving semidefinite programs
  via low-rank factorization}.
\newblock {\em Mathematical Programming}, 95(2):329--357, February 2003.

\bibitem{cai2019nonconvex}
Changxiao Cai, Gen Li, H~Vincent Poor, and Yuxin Chen.
\newblock Nonconvex low-rank symmetric tensor completion from noisy data.
\newblock {\em Advances in neural information processing systems}, 2019.

\bibitem{Candes:2008wp}
Emmanuel~J Cand{\`e}s.
\newblock {The restricted isometry property and its implications for compressed
  sensing}.
\newblock {\em Comptes Rendus Mathematique}, 346(9-10):589--592, May 2008.

\bibitem{Candes:2014br}
Emmanuel~J Cand{\`e}s and Carlos Fernandez-Granda.
\newblock {Towards a mathematical theory of super-resolution}.
\newblock {\em Communications on Pure and Applied Mathematics}, 67(6):906--956,
  June 2014.

\bibitem{Candes:2009kja}
Emmanuel~J Cand{\`e}s and Benjamin Recht.
\newblock {Exact matrix completion via convex optimization}.
\newblock {\em Foundations of Computational Mathematics}, 9(6):717--772, 2009.

\bibitem{Cardoso:1989eo}
Jean-Francois Cardoso.
\newblock {Source separation using higher order moments}.
\newblock {\em International Conference on Acoustics, Speech, and Signal
  Processing}, pages 2109--2112 vol.4, 1989.

\bibitem{Chandrasekaran:2010hl}
Venkat Chandrasekaran, Benjamin Recht, Pablo~A Parrilo, and Alan~S Willsky.
\newblock {The convex geometry of linear inverse problems}.
\newblock {\em Foundations of Computational Mathematics}, 12(6):805--849, 2012.

\bibitem{chi2014compressive}
Yuejie Chi and Yuxin Chen.
\newblock Compressive two-dimensional harmonic retrieval via atomic norm
  minimization.
\newblock {\em IEEE Transactions on Signal Processing}, 63(4):1030--1042, 2014.

\bibitem{chi2019nonconvex}
Yuejie Chi, Yue~M Lu, and Yuxin Chen.
\newblock Nonconvex optimization meets low-rank matrix factorization: An
  overview.
\newblock {\em IEEE Transactions on Signal Processing}, 67(20):5239--5269,
  2019.

\bibitem{comon:2009ur}
Pierre Comon.
\newblock {Tensor decompositions, state of the art and applications}.
\newblock {\em IMA Conf. Mathematics in Signal Processing}, May 2009.

\bibitem{davenport2016overview}
Mark~A Davenport and Justin Romberg.
\newblock An overview of low-rank matrix recovery from incomplete observations.
\newblock {\em IEEE Journal of Selected Topics in Signal Processing},
  10(4):608--622, 2016.

\bibitem{spark}
David~L Donoho and Michael Elad.
\newblock {Optimally sparse representation in general (nonorthogonal)
  dictionaries via 1 minimization}.
\newblock {\em Proceedings of the National Academy of Sciences},
  100(5):2197--2202, 2003.

\bibitem{eftekhari2019sparse}
Armin Eftekhari, Jared Tanner, Andrew Thompson, Bogdan Toader, and Hemant
  Tyagi.
\newblock Sparse non-negative super-resolution—simplified and stabilised.
\newblock {\em Applied and Computational Harmonic Analysis}, 2019.

\bibitem{FernandezGranda:2013wp}
Carlos Fernandez-Granda.
\newblock Support detection in super-resolution.
\newblock In {\em Proceedings of the 10th International Conference on Sampling
  Theory and Applications (SampTA 2013)}, pages 145--148, 2013.

\bibitem{fernandez2016demixing}
Carlos Fernandez-Granda, Gongguo Tang, Xiaodong Wang, and Le~Zheng.
\newblock {Demixing sines and spikes: Robust spectral super-resolution in the
  presence of outliers}.
\newblock {\em Information and Inference: A Journal of the IMA}, 2016.

\bibitem{foucart2013mathematical}
S.~Foucart and H.~Rauhut.
\newblock {\em {A mathematical introduction to compressive sensing}}.
\newblock Applied and Numerical Harmonic Analysis. Springer New York, 2013.

\bibitem{friedland2016nuclear}
Shmuel Friedland and Lek-Heng Lim.
\newblock {Nuclear norm of higher-order tensors}.
\newblock {\em Mathematics of Computation}, 2017.

\bibitem{Gandy:2011bh}
Silvia Gandy, Benjamin Recht, and Isao Yamada.
\newblock {Tensor completion and low-n-rank tensor recovery via convex
  optimization}.
\newblock {\em Inverse Problems}, 27(2):025010, February 2011.

\bibitem{ge2017learning}
Rong Ge, Jason~D Lee, and Tengyu Ma.
\newblock Learning one-hidden-layer neural networks with landscape design.
\newblock In {\em 6th International Conference on Learning Representations,
  ICLR 2018}, 2018.

\bibitem{Heckel:2014ts}
Reinhard Heckel, Veniamin~I Morgenshtern, and Mahdi Soltanolkotabi.
\newblock {Super-resolution radar}.
\newblock {\em Information and Inference: A Journal of the IMA}, 5(1):22--75,
  2016.

\bibitem{Hillar:2013by}
Christopher~J Hillar and Lek-Heng Lim.
\newblock {Most tensor problems are NP-Hard}.
\newblock {\em Journal of the ACM (JACM)}, 60(6):45--39, November 2013.

\bibitem{hopkins2019robust}
Samuel~B Hopkins, Tselil Schramm, and Jonathan Shi.
\newblock A robust spectral algorithm for overcomplete tensor decomposition.
\newblock In {\em Conference on Learning Theory}, pages 1683--1722. PMLR, 2019.

\bibitem{Chen:2005jn}
Jieqiong Hou and Haifeng Qian.
\newblock Collaboratively filtering malware infections: a tensor decomposition
  approach.
\newblock In {\em Proceedings of the ACM Turing 50th Celebration
  Conference-China}, page~28. ACM, 2017.

\bibitem{Huang:2014vc}
Bo~Huang, Cun Mu, Donald Goldfarb, and John Wright.
\newblock {Provable low-rank tensor recovery}.
\newblock {\em Optimization-Online}, 4252, 2014.

\bibitem{Jain:2014wm}
Prateek Jain and Sewoong Oh.
\newblock {Provable tensor factorization with missing data}.
\newblock In {\em Advances in Neural Information Processing Systems}, pages
  1431--1439, 2014.

\bibitem{kolda2009tensor}
Tamara~G Kolda and Brett~W Bader.
\newblock {Tensor decompositions and applications}.
\newblock {\em SIAM review}, 51(3):455--500, 2009.

\bibitem{Kruskal:1977dx}
Joseph~B Kruskal.
\newblock {Three-way arrays: rank and uniqueness of trilinear decompositions,
  with application to arithmetic complexity and statistics}.
\newblock {\em Linear Algebra and its Applications}, 18(2):95--138, January
  1977.

\bibitem{li2015overcomplete}
Qiuwei Li, Ashley Prater, Lixin Shen, and Gongguo Tang.
\newblock Overcomplete tensor decomposition via convex optimization.
\newblock In {\em 2015 IEEE 6th International Workshop on Computational
  Advances in Multi-Sensor Adaptive Processing (CAMSAP)}, pages 53--56. IEEE,
  2015.

\bibitem{li2017convex}
Qiuwei Li and Gongguo Tang.
\newblock Convex and nonconvex geometries of symmetric tensor factorization.
\newblock In {\em Asilomar Conference on Signals, Systems, and Computers},
  2017.

\bibitem{li2018approximate}
Qiuwei Li and Gongguo Tang.
\newblock Approximate support recovery of atomic line spectral estimation: A
  tale of resolution and precision.
\newblock {\em Applied and Computational Harmonic Analysis}, 2018.

\bibitem{li2017geometry}
Qiuwei Li, Zhihui Zhu, and Gongguo Tang.
\newblock Geometry of factored nuclear norm regularization.
\newblock {\em arXiv preprint arXiv:1704.01265}, 2017.

\bibitem{li2018non}
Qiuwei Li, Zhihui Zhu, and Gongguo Tang.
\newblock The non-convex geometry of low-rank matrix optimization.
\newblock {\em Information and Inference: A Journal of the IMA}, 8(1):51--96,
  2018.

\bibitem{li2019alternating}
Qiuwei Li, Zhihui Zhu, and Gongguo Tang.
\newblock Alternating minimizations converge to second-order optimal solutions.
\newblock In {\em International Conference on Machine Learning}, pages
  3935--3943. PMLR, 2019.

\bibitem{li2019provable}
Qiuwei Li, Zhihui Zhu, Gongguo Tang, and Michael~B Wakin.
\newblock Provable bregman-divergence based methods for nonconvex and
  non-lipschitz problems.
\newblock {\em arXiv preprint arXiv:1904.09712}, 2019.

\bibitem{li2022tensor}
Shuang Li and Qiuwei Li.
\newblock Local and global convergence of general burer-monteiro tensor
  optimizations.
\newblock In {\em Proceedings of the AAAI Conference on Artificial
  Intelligence}, volume~36, 2022.

\bibitem{li2020global}
Shuang Li, Qiuwei Li, Zhihui Zhu, Gongguo Tang, and Michael~B Wakin.
\newblock The global geometry of centralized and distributed low-rank matrix
  recovery without regularization.
\newblock {\em IEEE Signal Processing Letters}, 27:1400--1404, 2020.

\bibitem{lu2016tensor}
Canyi Lu, Jiashi Feng, Yudong Chen, Wei Liu, Zhouchen Lin, and Shuicheng Yan.
\newblock {Tensor robust principal component analysis: exact recovery of
  corrupted low-rank tensors via convex optimization}.
\newblock In {\em Proceedings of the IEEE Conference on Computer Vision and
  Pattern Recognition}, pages 5249--5257, 2016.

\bibitem{ma2016polynomial}
Tengyu Ma, Jonathan Shi, and David Steurer.
\newblock Polynomial-time tensor decompositions with sum-of-squares.
\newblock In {\em 2016 IEEE 57th Annual Symposium on Foundations of Computer
  Science (FOCS)}, pages 438--446. IEEE, 2016.

\bibitem{Mu:2013tz}
Cun Mu, Bo~Huang, John Wright, and Donald Goldfarb.
\newblock {Square deal: lower bounds and improved relaxations for tensor
  recovery}.
\newblock In {\em International Conference on Machine Learning}, pages 73--81,
  2014.

\bibitem{potechin2017exact}
Aaron Potechin and David Steurer.
\newblock Exact tensor completion with sum-of-squares.
\newblock In {\em Conference on Learning Theory}, pages 1619--1673. PMLR, 2017.

\bibitem{Recht:2011up}
Benjamin Recht.
\newblock {A simpler approach to matrix completion}.
\newblock {\em Journal of Machine Learning Research}, 12(Dec):3413--3430, 2011.

\bibitem{Recht:2010hta}
Benjamin Recht, Maryam Fazel, and Pablo~A Parrilo.
\newblock {Guaranteed minimum-rank solutions of linear matrix equations via
  nuclear norm minimization}.
\newblock {\em SIAM Review}, 52(3):471--501, August 2010.

\bibitem{Hsu:2013iz}
Hanie Sedghi, Majid Janzamin, and Anima Anandkumar.
\newblock Provable tensor methods for learning mixtures of generalized linear
  models.
\newblock In {\em Artificial Intelligence and Statistics}, pages 1223--1231,
  2016.

\bibitem{Lim:2010if}
Nicholas~D Sidiropoulos, Lieven De~Lathauwer, Xiao Fu, Kejun Huang, Evangelos~E
  Papalexakis, and Christos Faloutsos.
\newblock Tensor decomposition for signal processing and machine learning.
\newblock {\em IEEE Transactions on Signal Processing}, 65(13):3551--3582,
  2017.

\bibitem{Smilde:2005wf}
Age Smilde, Rasmus Bro, and Paul Geladi.
\newblock {\em {Multi-Way Analysis: Applications in the Chemical Sciences}}.
\newblock John Wiley {\&} Sons, 2005.

\bibitem{Tang:2013fo}
Gongguo Tang, B~N Bhaskar, P~Shah, and B~Recht.
\newblock {Compressed sensing off the grid}.
\newblock {\em Information Theory, IEEE Transactions on}, 59(11):7465--7490,
  2013.

\bibitem{Tang:2013gd}
Gongguo Tang, Badri~Narayan Bhaskar, and Benjamin Recht.
\newblock {Near minimax line spectral estimation}.
\newblock {\em IEEE Transactions on Information Theory}, 61(1):499--512, 2015.

\bibitem{Tang:2013:ad}
Gongguo Tang and Benjamin Recht.
\newblock {Atomic decomposition of mixtures of translation-invariant signals}.
\newblock In {\em IEEE International Workshop on Computational Advances in
  Multi-Sensor Adaptive Processing CAMSAP}, Saint Martin, December 2013.

\bibitem{Tang:2015gt}
Gongguo Tang and Parikshit Shah.
\newblock {Guaranteed tensor decomposition: a moment approach}.
\newblock In {\em International Conference on Machine Learning}, Lille, France,
  2015.

\bibitem{Tang:2014jk}
Gongguo Tang, Parikshit Shah, Badri~Narayan Bhaskar, and Benjamin Recht.
\newblock {Robust line spectral estimation}.
\newblock In {\em 2014 48th Asilomar Conference on Signals, Systems and
  Computers}, pages 301--305. IEEE, 2014.

\bibitem{Watson:1992ha}
G~A Watson.
\newblock {Characterization of the subdifferential of some matrix norms}.
\newblock {\em Linear Algebra and its Applications}, 170:33--45, June 1992.

\bibitem{yuan2016tensor}
Ming Yuan and Cun-Hui Zhang.
\newblock On tensor completion via nuclear norm minimization.
\newblock {\em Foundations of Computational Mathematics}, 16(4):1031--1068,
  2016.

\bibitem{zhu2018global}
Zhihui Zhu, Qiuwei Li, Gongguo Tang, and Michael~B Wakin.
\newblock Global optimality in low-rank matrix optimization.
\newblock {\em IEEE Transactions on Signal Processing}, 66(13):3614--3628,
  2018.

\bibitem{zhu2021global}
Zhihui Zhu, Qiuwei Li, Gongguo Tang, and Michael~B Wakin.
\newblock The global optimization geometry of low-rank matrix optimization.
\newblock {\em IEEE Transactions on Information Theory}, 67(2):1308--1331,
  2021.

\bibitem{zhu2019distributed}
Zhihui Zhu, Qiuwei Li, Xinshuo Yang, Gongguo Tang, and Michael~B Wakin.
\newblock Distributed low-rank matrix factorization with exact consensus.
\newblock {\em Advances in Neural Information Processing Systems},
  32:8422--8432, 2019.

\end{thebibliography}

\end{document}